\documentclass[onecolumn,draftcls]{IEEEtran}
\usepackage{cite}
\usepackage[cmex10]{amsmath}
\usepackage{graphicx,psfig,amsmath} 
\usepackage{algorithm, algpseudocode}
\usepackage{amssymb, theorem} 
\usepackage{bbm}
\usepackage[dvips]{changebar}
\theoremheaderfont{\scshape} \theoremstyle{plain}
\newtheorem{thm}{Theorem}[section]

\newtheorem{lem}[thm]{Lemma}

\newcommand{\R}{\mathbb{R}}

\usepackage[footnotesize]{subfigure}
\usepackage{url}

\begin{document}

\title{An Inflationary Differential Evolution Algorithm for Space Trajectory Optimization}

\author{Massimiliano Vasile,~\IEEEmembership{Member,~IEEE},
        Edmondo Minisci, and
        Marco Locatelli,~\IEEEmembership{Member,~IEEE}

\thanks{M. Vasile is with the Department of Mechanical Engineering, University of Strathclyde, James Weir Building
75 Montrose Street G1 1XJ
(phone: +44-(0)141-548-4851;
fax: +44-(0)141-552-5105; email: massimiliano.vasile@strath.ac.uk).}
\thanks{E. Minisci is with the School of Engineering, University of Glasgow, James
Watt South Building, G12 8QQ, Glasgow, UK (phone: +44-141-330-8470;
fax: +44-141-330-5560; email: edmondo.minisci@glasgow.ac.uk).}
\thanks{M. Locatelli is with the Dipartimento di Ingegneria dell'Informazione, Universit\'{a} degli Studi di Parma, Parma, via G. P. Usberti, 181/A, 43124, Italy
(email: locatell@ce.unipr.it).}
\thanks{Manuscript received \today; revised \today.}}

\markboth{Journal of \LaTeX\ Class Files,~Vol.~6, No.~1, January~2007}%
{Shell \MakeLowercase{\textit{et al.}}: Bare Demo of IEEEtran.cls for Journals}

\maketitle

\begin{abstract}
In this paper we define a discrete dynamical system that governs the evolution of a population of agents. From the dynamical system, a variant of Differential Evolution is derived. It is then demonstrated
that, under some assumptions on the differential mutation strategy
and on the local structure of the objective function, the proposed
dynamical system has fixed points towards which it converges
with probability one for an infinite number of generations. This property is used to derive an algorithm that
performs better than standard Differential Evolution on some space
trajectory optimization problems. The novel algorithm is then
extended with a guided restart procedure that further increases the
performance, reducing the probability of stagnation in deceptive
local minima.
\end{abstract}


\IEEEpeerreviewmaketitle

\section{Introduction}
\IEEEPARstart{S}{ome} evolutionary heuristics can be interpreted as
discrete dynamical systems governing the movements of a set of
agents (or particles) in the search space. This is well known for
Particle Swarm Optimization (PSO) where the variation of the
velocity of each particle in the swarm is defined by a control term
made of a social component and a individual (or cognitive) component
\cite{clerc:02, trlea:03, clerc:06, poli:08}. The social component
can be interpreted as a behavior dictated by the knowledge acquired
by the whole swarm of particles, while the cognitive component can
be interpreted as a behavior dictated by the knowledge acquired by
each individual particle.

The same principle can be generalized and extended to other
evolutionary heuristics such as Differential Evolution
(DE) \cite{kenneth:05}. The analysis of a discrete dynamical system
governed by the heuristics generating the behavior of individuals in
DE, and in Evolutionary Computation in general, would allow for a
number of considerations on the evolution of the search process and
therefore on the convergence properties of a global optimization algorithm. In particular, four outcomes of the search
process are possible:
\begin{itemize}
  \item Divergence to infinity. In this case the discrete dynamical
  system is unstable, and the global optimization algorithm is not convergent.
  \item Convergence to a fixed point in $D$. In this case the global
  optimization algorithm is simply convergent in $D$ and we can
  define a stopping criterion. Once the search is stopped we can
  define a restart procedure. Depending on the convergence profile,
  the use of a restart procedure can be more or less efficient.
  \item Convergence to a limit cycle in which the same points in $D$
  are re-sampled periodically. Even in this case we can define a
  stopping criterion and a restart procedure.
    \item Convergence to a strange (chaotic) attractor. In this case a
    stopping criterion cannot be clearly defined because different
    points are sampled at different iterations.
\end{itemize}
All outcomes are generally interesting to understand the evolution process.
Identifying under which conditions divergence occurs is important to properly design an algorithm, or to define the appropriate setting, in particular in the case of automatic adaptation of some key parameters. Divergence can be seen as the opposite of the intensification process and can increase diversity and exploration.
The convergence to a chaotic attractor can represent an interesting case of dense random sampling of an extended region. When this happens, the algorithm repeatedly samples the same region but never re-samples the same points. This mechanism could be useful to reconstruct an extended set with a specific property (e.g. $f(x)<\varepsilon$, with $f$ the objective function and $\varepsilon$ a threshold).

The main interest, in this paper, is in the most commonly desired outcome: the convergence to a fixed point
(see for example \cite{Dellnitz02}, in the context of global root finding, for methods to eliminate undesirable behaviors). The convergence to a fixed point, even different from an optimum, can be used to induce a restart of the search process with minimum waste of resources (as it will be demonstrated in this paper).

The analysis of the dynamical properties of the
dynamical system associated with a particular heuristic can give
some insights into the balance between global and local exploration
and the volume of the search space that is covered during the
search. Extensive work on the dynamics of Genetic Algorithms and
general Evolutionary Algorithms can be found in the studies of
Pr\"{u}gel-Bennett et al. \cite{bennett:97, bennett:01} and Beyer
\cite{reeves:98}. 
Non-evolutionary examples of the use of
dynamical system theory to derive effective global optimization
schemes can be found in the works of Sertl and Dellnitz
\cite{dell:06}.

In this paper we propose a discrete
dynamical system, or discrete map, governing the evolution of a population of individuals, being each individual a point in a $d$-dimensional domain $D$. From the discrete dynamical system we derive a variant of Differential
Evolution and we study its convergence properties. In particular, it is proven that, under some assumptions,
the dynamics of the proposed variant of DE converges to a fixed point or to a level set.

Note that, the
proofs proposed in this paper considers the whole $d$-dimensional
discrete dynamical system associated to the evolution of the
population.

The theoretical results are then used to derive a novel algorithm that
performs better than DE strategies \emph{DE/rand/1/bin} and \emph{DE/best/1/bin}\cite{storn97} on some difficult space trajectory design
problems. The novel algorithm is based on a hybridization of the proposed DE variant and the logic behind Monotonic Basin Hopping (MBH) \cite{Lea:00}.

The paper is organized as follows: after introducing the dynamics of
a population of agents, three convergence theorems are demonstrated.
Then an inflationary DE algorithm based on a hybridization between a variant of DE and MBH is derived. A description of the test
cases and  testing procedure follows. After presenting the
results of the tests, the search space is analyzed to derive some
considerations on the general applicability of the results.



\section{Agent Dynamics}\label{sec:agent_dyna}
In this section we start by defining a generic discrete dynamical system governing the motion of an agent in a generic search space $D$. From this general description we derive a variant of what Storn and Price defined as the basic DE strategy\cite{storn97}.

If we consider that a candidate solution vector in a generic
$d$-dimensional box,
\begin{equation}
D=[b_l(1)\; b_u(1)]\times[b_l(2)\;b_u(2)]\times...[b_l(d)\; b_u(d)]
\end{equation}
is associated to an agent, then the heuristic governing its motion in $D$ can be written as the following discrete dynamical system:
\begin{equation}\label{eq:ma_dyn}
\begin{array}{l}
  \mathbf{v}_{i,k+1}=(1-c)\mathbf{v}_{i,k}+\mathbf{u}_{i,k} \\
  \mathbf{x}_{i,k+1}=\mathbf{x}_{i,k}+\nu S(\mathbf{x}_{i,k}+\mathbf{v}_{i,k+1},\mathbf{x}_{i,k}) \mathbf{v}_{i,k+1}
\end{array}
\end{equation}
with
\begin{equation}\label{eq:v_limit}
\nu=\min\left(\left[v_{max}, v_{i,k+1}\right]\right)/v_{i,k+1}
\end{equation}

The function $S(\mathbf{x}_{i,k}+\mathbf{v}_{i,k+1},\mathbf{x}_{i,k})$ is a selection operator that can be
either $1$ if the candidate point $\mathbf{x}_{i,k}+\mathbf{v}_{i,k+1}$ is accepted or $0$ if it is not
accepted. Different evolutionary algorithms have different ways to define $S$.

The control $\mathbf{u}_{i,k}$ defines the next point that will be
sampled for each one of the existing points in the solution space,
the vectors $\mathbf{x}_{i,k}$ and $\mathbf{v}_{i,k}$ define the
current state of the agent in the solution space at stage $k$ of the
search process, $\mathbf{x}_{i,k+1}$ and $\mathbf{v}_{i,k+1}$ define the
state of the agent in the solution space at stage $k+1$ of the
search process and $c$ is a viscosity, or dissipative coefficient,
for the process. Eq. \eqref{eq:v_limit} represents a limit sphere
around the agent $\mathbf{x}_{i,k}$ at stage $k$ of the search
process. Different evolutionary algorithms have different ways to define $\mathbf{u}_{i,k}$, $\nu$ and $c$ (see for example the dynamics of PSO \cite{clerc:02}). In this paper we will focus on the way $\mathbf{u}_{i,k}$, $\nu$ and $c$ are defined in Differential Evolution.

Now consider the $\mathbf{u}_{i,k}$ defined by:
\begin{equation}\label{eq:de}
    \mathbf{u}_{i,k}=\mathbf{e} \left[ (\mathbf{x}_{i_3,k}-\mathbf{x}_{i,k})+F(\mathbf{x}_{i_2,k}-\mathbf{x}_{i_1,k})\right]
\end{equation}
where $i_1$, $i_2$ and $i_3$ are integer numbers randomly chosen in the
interval $[1, \; n_{pop}]\subset\mathbb{N}$ of indexes of the
population, and $\mathbf{e}$ is a mask containing a random number of
$0$ and $1$ according to:
\begin{equation}\label{eq:e_def}
    e(j)= \Bigg\lbrace
           \begin{array}{c}
             1\Rightarrow r_j\leq C_R \\
             0\Rightarrow r_j> C_R \\
           \end{array}
\end{equation}
with $j$ randomly chosen in the interval $[1 \; d]$, so that \textbf{e} contains at least one nonzero component, $d$ the dimension of the solution vector, $r_j$
a random number taken from a random uniform distribution $r_j\in U[0,1]$ and $C_R$
a constant. The product between $\mathbf{e}$ and $\left[ (\mathbf{x}_{i_3,k}-\mathbf{x}_{i,k})+F(\mathbf{x}_{i_2,k}-\mathbf{x}_{i_1,k})\right]$ in Eq. \eqref{eq:de} has to be intended component-wise. The index $i_3$ can be chosen at random (this option will be called exploration
strategy in the remainder of this paper) or can be the index of the best solution vector
$\mathbf{x}_{best}$ (this option will be called convergence strategy in the remainder of this paper). Selecting the best
solution vector or a random one changes significantly the
convergence speed of the algorithm. The selection function $S$ can be either 1 or 0
depending on the relative value of the objective function of the new
candidate individual generated with Eq. \eqref{eq:de} with respect to
the one of the current individual.
\\
In other words the selection function $S$ can be expressed as:
\begin{equation}\label{eq:DE_selection}
S(\mathbf{x}_{i,k}+\mathbf{u}_{i,k},\mathbf{x}_{i,k})=\Big \{ \begin{array}{l}
  1 \;\; \text{if} \;\; f(\mathbf{x}_{i,k}+\mathbf{u}_{i,k})<f(\mathbf{x}_{i,k})\\
  0 \;\; \text{otherwise}
\end{array}
\end{equation}
Note that in this paper we consider minimization problems in which the lowest value of $f$ is sought.

If one takes $c=1$, $\nu=1$ and
$v_{max}=+\infty$, then map \eqref{eq:ma_dyn} reduces to:
\begin{eqnarray}\label{eq:ma_dyn_DE}
  \mathbf{x}_{i,k+1}=\mathbf{x}_{i,k}+S(\mathbf{x}_{i,k}+\mathbf{u}_{i,k},\mathbf{x}_{i,k}) \mathbf{u}_{i,k}
\end{eqnarray}

In the general case the indices
$i_{1}$, $i_{2}$ and $i_{3}$ can assume any value.
However, if the three indexes $i_{1}$, $i_{2}$ and $i_{3}$ are restricted to be mutually different, map \eqref{eq:ma_dyn_DE}, with $\mathbf{u}_{i,k}$ defined in \eqref{eq:de}, \textbf{e} defined in \eqref{eq:e_def} and $S$ defined in \eqref{eq:DE_selection}, is the DE basic strategy \emph{DE/rand/1/bin} defined by Storn and Price in \cite{storn97}. If $i_3$ is taken as the index of the best individual $i_{best}$, and $i_{1}$, $i_{2}$ and $i_{3}$ are mutually different, then one can obtain the DE strategy \emph{DE/best/1/bin}. In fact, if \eqref{eq:de} is substituted in \eqref{eq:ma_dyn_DE} one gets:
\begin{eqnarray}\label{eq:standard_DE}
  \mathbf{x}_{i,k+1}=(\mathbf{1}-S\mathbf{e})\mathbf{x}_{i,k} + S\mathbf{e} \left[ (\mathbf{x}_{i_3,k}+F(\mathbf{x}_{i_2,k}-\mathbf{x}_{i_1,k})\right]
\end{eqnarray}

Now, if the selection function does not accept the candidate point $\mathbf{x}_{i,k}+\mathbf{u}_{i,k}$, then $S=0$, therefore \eqref{eq:standard_DE} reduces to $\mathbf{x}_{i,k+1}=\mathbf{x}_{i,k}$ (i.e. the state of the agent remains unchanged). If the candidate point is accepted instead, then the new location of the agent is:
\begin{eqnarray}\label{eq:standard_DE2}
  \mathbf{x}_{i,k+1}=(\mathbf{1}-\mathbf{e})\mathbf{x}_{i,k} + \mathbf{e} \left[ (\mathbf{x}_{i_3,k}+F(\mathbf{x}_{i_2,k}-\mathbf{x}_{i_1,k})\right]
\end{eqnarray}
The quantity in square brackets is the basic DE strategy \emph{DE/rand/1/bin} or the variant \emph{DE/best/1/bin}, respectively for $i_3$ random or $i_3=i_{best}$. The mask $\mathbf{e}$ together with $(\mathbf{1}-\mathbf{e})\mathbf{x}_{i,k}$ represent the cross-over operator in the basic DE strategy \cite{storn97}, with $\mathbf{1}$ a vector of $1$'s and the products $\mathbf{1}\mathbf{x}_{i,k}$ and $\mathbf{e}\mathbf{x}_{i,k}$ that are both component-wise. In fact, $\mathbf{x}_{i,k+1}$ is made of the components of
$\mathbf{x}_{i,k}$ that correspond to the zero elements of \textbf{e}, and the components of $\left[ (\mathbf{x}_{i_3,k}+F(\mathbf{x}_{i_2,k}-\mathbf{x}_{i_1,k})\right]$ that correspond to the nonzero elements of \textbf{e}.
If the assumption of mutually different indexes is dropped then map \eqref{eq:ma_dyn_DE} can be seen as a further variant of the basic DE strategy.

In compact matrix form for the entire population, map \eqref{eq:ma_dyn_DE} can be written as:
\begin{equation}\label{eq:ma_dyn_DE_matrix}
\mathbf{X}_{k+1}=\mathbf{J}_k \mathbf{X}_{k}
\end{equation}
with the $i$-th line of matrix $\mathbf{X}_k\in \R^{n_{pop}\times
d}$ is point $\mathbf{x}_{i,k}$.

To be more precise, in the case
$\mathbf{x}_{i,k+1}\not \in D$, then every component $j$ violating
the boundaries $\mathbf{b}_l$ and $\mathbf{b}_u$ is projected back
into $D$ by picking the new value
$x(j)_{i,k+1}=b_l(j)+\varsigma(b_u(j)-b_l(j))$, where $\varsigma$
is taken from a random uniform distribution $\varsigma \in U[0,1]$.  The interest is now in
the properties of map \eqref{eq:ma_dyn_DE_matrix}. We start by
observing that if $S(\mathbf{x}_{i,k}+\mathbf{u}_{i,k},\mathbf{x}_{i,k})=1\Leftrightarrow
f(\mathbf{x}_{i,k}+\mathbf{u}_{i,k})<f(\mathbf{x}_{i,k})$, for $i=1,...,n_{pop}$, the global minimizer
$\mathbf{x}_{g}\in D$ is a fixed point for map \eqref{eq:ma_dyn_DE}
since every point $\mathbf{x}\in D$ is such that
$f(\mathbf{x})\geq f(\mathbf{x}_g)$.

Then, let us assume that at every iteration $k$ we can find two
connected subsets $D_k$ and $D_k^*$ of $D$ such that
$f(\mathbf{x}_k)<f(\mathbf{x}_k^*), \forall \mathbf{x}_k\in
D_k,\forall \mathbf{x}_k^*\in D_k^*\setminus D_k$, and let us also
assume that $P_k\subseteq D_k$ while $P_{k+1}\subseteq D_k^*$
(recall that $P_k$ and $P_{k+1}$ denote the populations at iteration
$k$ and $k+1$ respectively). If $\mathbf{x}_l$ is the lowest local
minimum in $D_k$, then $\mathbf{x}_l$ is a fixed point in $D_k$ for
map \eqref{eq:ma_dyn_DE}. In fact, every point generated by map \eqref{eq:ma_dyn_DE} must be in $D_k$ and
$f(\mathbf{x}_l)<f(\mathbf{x}),\forall \mathbf{x}\in D_k$.

Moreover, under the above assumptions the reciprocal distance of
the individuals cannot grow indefinitely because of the map
\eqref{eq:ma_dyn_DE}, therefore the map cannot be divergent.

Finally, a matrix $\mathbf{X}_k$ whose lines are
$n_{pop}$ replications of a unique point $\mathbf{x}\in D$, is
always a fixed point for the map \eqref{eq:ma_dyn_DE_matrix}.

Now one can consider two variants of Differential Evolution: one in which
index $i_{1}$ can be equal to index $i_{2}$ but both are different from $i_{3}$ and one in which index $i_{1}$ can be equal to index $i_{3}$ but both are different from $i_{2}$. In these cases two interesting results can be proven. First of all we prove that if $i_{1}$ can be equal to index $i_{2}$ then the population can collapse to a single point in $D$.

\begin{thm}
\label{thm:collapse_map}

If, for every $k$, $i_{1}$ is equal to $i_{2}$ with strictly positive probability and $f(\mathbf{x}_{i_{best},k})=f(\mathbf{x}_{i,k}) \Leftrightarrow \mathbf{x}_{i_{best},k}=\mathbf{x}_{i,k}$(i.e., the set of best points within the population is made up by a single point, multiple copies of which possibly exist),then the population collapses to a single point with probability 1 for $k\rightarrow\infty$ under the effect of the discrete dynamical system (\ref{eq:ma_dyn_DE})
\end{thm}

\begin{proof}
If $i_{1}$ is equal to $i_{2}$ with strictly positive probability, and $f(\mathbf{x}_{i_{best},k})=f(\mathbf{x}_{i,k}) \Leftrightarrow \mathbf{x}_{i_{best},k}=\mathbf{x}_{i,k}$ for every $k$, then map (\ref{eq:ma_dyn_DE}) at each iteration $k$ can generate,
with strictly positive probability, a displacement
$\mathbf{u}_{i,k}=(\mathbf{x}_{i_{3},k}-\mathbf{x}_{i,k})$ for each member $i$, $i=1,\ldots,n_{pop}$.
This happens if, for each $i\in [1,\ldots,n_{pop}]$, the following event, whose probability is strictly positive,
occurs
$$
e(s)=1,\ s=1,\ldots,d,\ \ i_3\neq i_1=i_2.
$$

\begin{equation}
\label{eq:indexi3}
i_3\in \{i\ :\ f(x_{i,k})=f(x_{i_{best},k})\} \end{equation}

Then, at each iteration we have a strictly
positive probability that the two or more individuals collapse into the single point $x_{i_{best},k}$ and for $k\rightarrow\infty$ the whole population
collapses to a single point with probability one.
\end{proof}

After a collapse, the population cannot progress further and needs to be restarted.
It is important to evaluate the probability of a total collapse, in fact if the collapse is progressive the population can keep on exploring but if the collapse is instantaneous the evolution ceases. Let us analyze the worst case in which, $i_3=i_{best}$. Then if for all the individuals $i_1=i_2$ and $\mathbf{e}=1$ we have a total instantaneous collapse. The probability of having $i_1=i_2$
is $1/n_{pop}$ and the probability of having $n_{pop}-1$ individuals collapsing at the same iteration $k$ is $(1/n_{pop})^{n_{pop}-1}$. Furthermore, given a $C_R\neq 1$ the probability to have $\mathbf{e}=1$ is $C_R^{d-1}$. Thus the total probability of an instantaneous total collapse is $[(1/n_{pop}) C_R^{d-1}]^{n_{pop}-1}$. The event has positive but small probability to happen. The complementary probability is $1-[(1/n_{pop}) C_R^{d-1}]^{n_{pop}-1}$ and the probability to have at least one total collapse after $k_h$ generations is $1-\{1-[(1/n_{pop}) C_R^{d-1}]^{n_{pop}-1}\}^{k_h}$.
Therefore, allowing the indexes to assume the value $i_1=i_2$ introduces the following interesting property. If the population is stagnating, and the condition $f(\mathbf{x}_{i_{best},k})=f(\mathbf{x}_{i,k}) \Leftrightarrow \mathbf{x}_{i_{best},k}=\mathbf{x}_{i,k}$ holds true, eventually there will be a total collapse and the population can be restarted with no risk to interrupt the evolutionary process.

If $i_{1}$ is equal to $i_{3}$ with strictly
positive probability but both are always different from $i_{2}$, then convergence to a fixed point can
be guaranteed if the function $f$ is strictly quasi-convex \cite{convex}
in $D$, and $D$ is compact and convex. In other words, under the
given assumptions, the population will converge to a single point.
We immediately remark that such point is not necessarily a local minimizer of the problem.
\begin{lem}
\label{lem:delta_min} Let $f$ be a continuous and strictly quasi-convex function
on a set $D$ and let us assume that $D$ is compact, convex and is not a singleton. Then, the following minimization problem with
$F\in (0,1)$ has a strictly positive minimum value
$\delta_r(\epsilon)$ for $\epsilon$ small enough:

\begin{equation}\label{eq:min_d}
  \begin{array}{lll}
    \delta_r(\epsilon)=&\min & g(\mathbf{y}_1,\mathbf{y}_2)=f(\mathbf{y}_2)-f(F\mathbf{y}_1+(1-F)\mathbf{y}_2)\\
    &s.t. & \mathbf{y}_1,\mathbf{y}_2 \in D \\
    & & \|\mathbf{y}_1-\mathbf{y}_2\|\geq \epsilon \\
    & & f(\mathbf{y}_1)\leq f(\mathbf{y}_2) \\
  \end{array}
\end{equation}
\end{lem}
\begin{proof}

Since $f$ is strictly quasi-convex $g(\mathbf{y}_1,\
\mathbf{y}_2)>0$, $\forall \mathbf{y}_1,\  \mathbf{y}_2 \in D$;
furthermore, the feasible region is nonempty (if $\epsilon$ is small enough and $D$ is not a singleton) and compact. Therefore,
according to Weierstrass' theorem the function $g$ attains its
minimum value over the feasible region. If we denote by
$(\mathbf{y}^*_1,\mathbf{y}^*_2)$ a global minimum point of the
problem, then we have
\begin{equation}
\label{eq:deltaeps}
\delta_r(\epsilon)=g(\mathbf{y}^*_1,\mathbf{y}^*_2)>0.
\end{equation}
\end{proof}
\begin{thm}
\label{thm:fixed_map} Assume that index $i_1$ can be equal to $i_3$. Given a function $f$ that is strictly
quasi-convex over the compact and convex set $D$, and a population $P_k\in D$, then if
$F\in(0,1)$ and $S(\mathbf{x}_{i,k}+\mathbf{u}_{i,k},\mathbf{x}_{i,k})=1\Leftrightarrow
f(\mathbf{x}_{i,k}+\mathbf{u}_{i,k})<f(\mathbf{x}_{i,k})$, for $i=1,...,n_{pop}$ , the population $P_k$
converges to a single point in $D$ for $k\rightarrow \infty$ with probability one.
\end{thm}
\begin{proof}
By contradiction let us assume that we do not have convergence to a
fixed point. Then, it must hold that:
\begin{equation}\label{eq:false_stat}
    \inf_k \max \big\{ \| \mathbf{x}_{i,k}-\mathbf{x}_{j,k}\|,i,j\in[1,...,n_{pop}] \big\} \geq
    \epsilon >0
\end{equation}
At every generation $k$ the map can generate with a strictly
positive probability, a displacement
$F(\mathbf{x}_{i^*,k}-\mathbf{x}_{j^*,k})$ for $\mathbf{x}_{j^*,k}$, where $i^*$ and $j^*$
identify the individuals with the maximal reciprocal distance, such
that the candidate point is
$\mathbf{x}_{cand}=F\mathbf{x}_{i^*,k}+(1-F)\mathbf{x}_{j^*,k}$ with
$f(\mathbf{x}_{i^*,k})\leq f(\mathbf{x}_{j^*,k})$. Since the
function $f$ is strictly quasi-convex, the candidate point is
certainly better than $\mathbf{x}_{j^*,k}$ and, therefore, is
accepted by $S$. Now, in view of Eq. \eqref{eq:false_stat} and of Lemma
\eqref{lem:delta_min} we must have that,
\begin{equation}\label{eq:f_red}
f(\mathbf{x}_{cand})\leq f(\mathbf{x}_{j,k})-\delta_r(\epsilon)
\end{equation}
Such reduction will occur with probability one infinitely often, and
consequently the function value of at least one individual will be,
with probability one, infinitely often reduced by
$\delta_r(\epsilon)$. But in this way the value of the objective
function of such individual would diverge to $-\infty$, which is a
contradiction because $f$ is bounded from below over the compact set
$D$.
\end{proof}

If a local minimum satisfies some regularity assumptions (e.g., the Hessian at
the local minimum is definite positive), then we can always define a neighborhood
such that: (i) map (\ref{eq:ma_dyn_DE}) will be unable to accept points outside the neighborhood;
(ii) the function is strictly convex within the neighborhood.
Therefore, if at some iteration $k$ the population $P_k$ belongs to such a neighborhood,
we can guarantee that map (\ref{eq:ma_dyn_DE_matrix}) will certainly
converge to a fixed point made up by $n_{pop}$ replications of a single point belonging to the neighborhood.
 For general functions, we can not always
guarantee that the population will converge to a fixed point, but we
can show that the maximum difference between the objective function
values in the population converges to 0, i.e. the points in the
population tend to belong to the same level set.
The proof is closely related to that of Theorem II.1 and we still need to assume that indices $i_1$ and $i_2$ can be equal.

\begin{thm}\label{thm:fixed_level}
Assume that index $i_1$ can be equal to $i_2$. Given a function $f$, limited from below over $D$, and a population
$P_k\in D$, then if $F\in(0,1)$ and
$S(\mathbf{x}_{i,k}+\mathbf{u}_{i,k},\mathbf{x}_{i,k})=1\Leftrightarrow
f(\mathbf{x}_{i,k}+\mathbf{u}_{i,k})<f(\mathbf{x}_{i,k})$, for $i=1,...,n_{pop}$ , the following holds
\begin{equation}
\max_{i,j\in [1,...,n_{pop}]}\ \mid
f(\mathbf{x}_{j,k})-f(\mathbf{x}_{i,k})\mid \rightarrow 0,
\end{equation}
as $k\rightarrow \infty$ with probability one.
\end{thm}
\begin{proof}
Let $S_k^*$ denote the set of best points in population $P_k$, i.e.
\begin{equation}
S_k^*=\big\{ \mathbf{x}_{j,k}\ :\  f(\mathbf{x}_{j,k})\leq
f(\mathbf{x}_{i,k})\ \ \forall\ i\in [1,\ldots,n_{pop}]  \big\}
\end{equation}
At each iteration $k$ there is a strictly positive probability that
the whole population will be reduced to $S_k^*$ at the next
iteration. To show this it is enough to substitute the condition stated in
(\ref{eq:indexi3}) with the following condition $$ i_3\in \{i\ :\ x_{i,k}\in S_k^*\} $$ (basically, with respect to Theorem II.1, we only drop the requirement that at each iteration the best value of the population is attained at a single point).

In other words, there is a strictly positive probability
for the event that the population at a given iteration will be made
up of points all with the same objective function value. Therefore,
such an event will occur infinitely often with probability one. Let us
denote with $\{k_h\}_{h=1,\ldots}$ the infinite subsequence of
iterations at which the event is true, and let
\begin{equation}
\Delta_h= f(\mathbf{x}_{i,k_{h}})-f(\mathbf{x}_{i,k_{h+1}})
\end{equation}
be the difference in the objective function values at two
consecutive iterations $k_h$ and $k_{h+1}$ (note that, since at
iterations $k_h$, $h=1,\ldots$ the objective function values are all
equal and any $i$ can be employed in the above definition). It holds
that for all $i,j\in [1,\ldots, n_{pop}]$
\begin{equation}
\mid f(\mathbf{x}_{j,k})-f(\mathbf{x}_{i,k}) \mid \,\,  \leq \Delta_h\ \ \
\forall\ k\in [k_h,k_{h+1}]
\end{equation}
Therefore, if we are able to prove that $\Delta_h\rightarrow 0$, as
$h\rightarrow \infty$, then we can also prove the result of our
theorem. Let us assume by contradiction that
$\Delta_h\not\rightarrow 0$. Then, there will exist a $\delta>0$
such that $\Delta_h\geq \delta$ infinitely many times. But this
would lead to function values diverging to $-\infty$ and,
consequently, to a contradiction.
\end{proof}

As a consequence of these results, for the choice of the index $i_1$, $i_2$ and $i_3$ non-mutually different, a possible stopping criterion for the
dynamics in Eq. \eqref{eq:ma_dyn_DE} would be to stop when the difference
between the function values in the population drops below a given
threshold. However, this could cause a premature halt of the evolutionary process. Indeed, even if at some iteration the function value at all points of the population is equal, this does not necessarily mean that the algorithm will be unable to make
further progress (unless all points in the population are multiple copies of a unique point). Therefore, since the evolution definitely ceases when the population
contracts to a single point, we can alternatively use as a
stopping criterion the fact that the maximum distance between points
in the population drops below a given threshold.

It is important to observe that the contraction of the population does not depend on whether the function $f$ is minimized or maximized but depends only on the definition of $S$ and on whether the function is bounded or unbounded.

\begin{figure}[tbh]
\begin{center}
\subfigure[]{ \label{fig:convx}
\includegraphics[width=8.0cm]{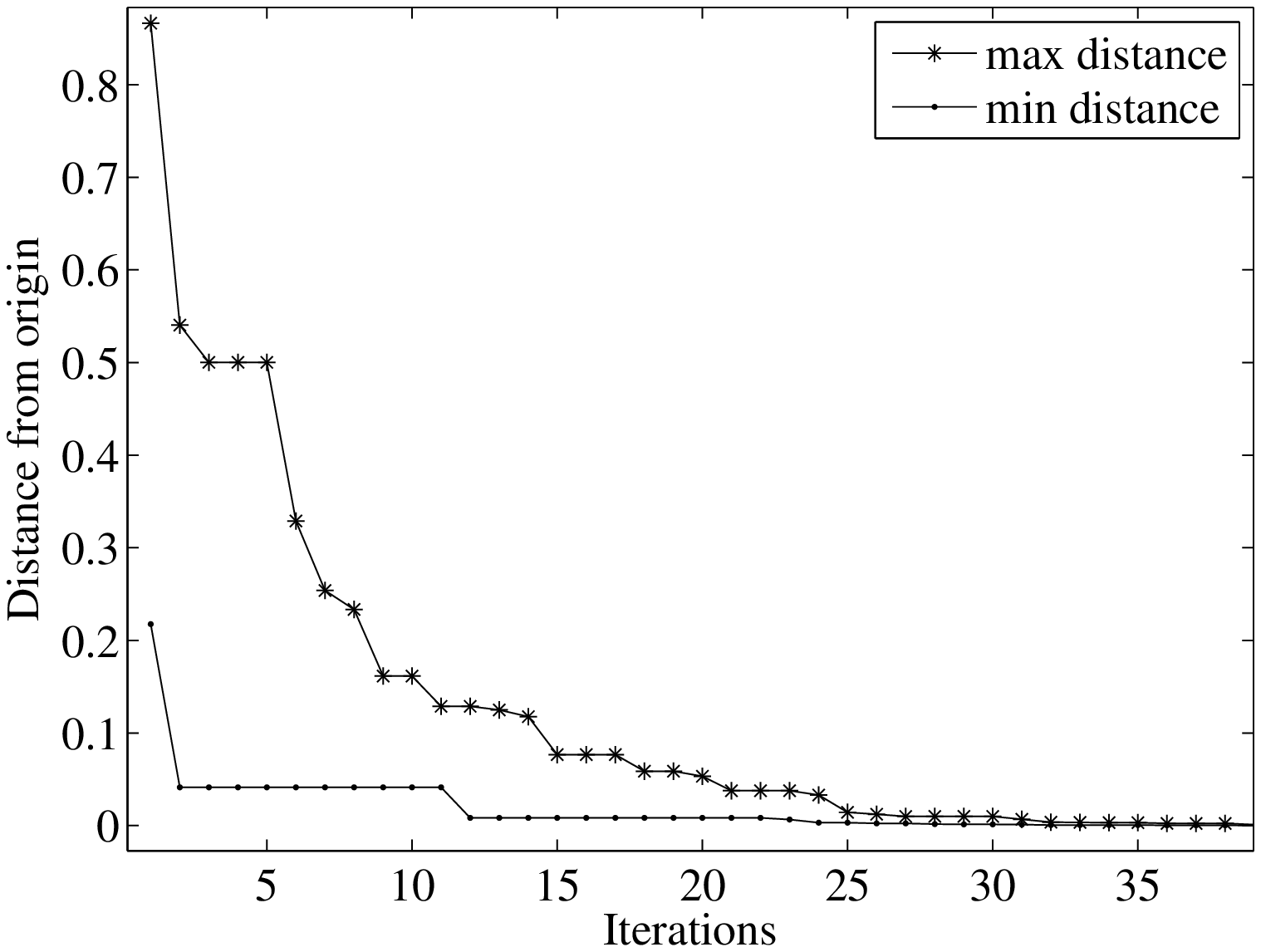}}
\subfigure[]{ \label{fig:detJ}
\includegraphics[width=8.0cm]{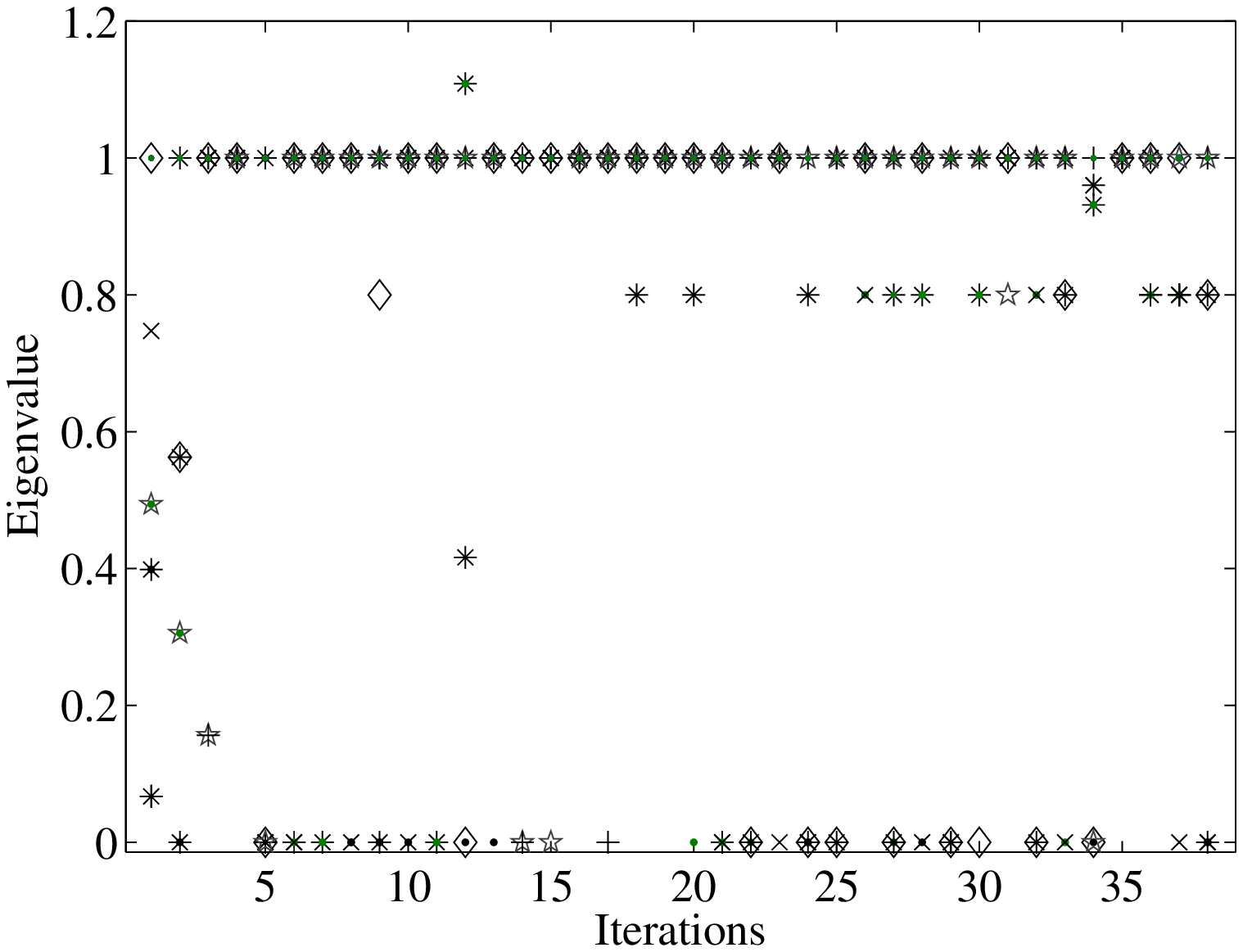}}
\caption{Contraction of map (\ref{eq:ma_dyn_DE_matrix}): a) max and
min distance of the individuals in the population from the origin,
b) eigenvalues with the number of evolutionary iterations.}
\label{fig:DE_Diss}
\end{center}
\end{figure}

To further verify the contraction properties of the dynamics in Eq.
\eqref{eq:ma_dyn_DE} one can look at the eigenvalues of the matrix
$\mathbf{J}_k$.

\begin{algorithm}[!t]
\caption {Inflationary Differential Evolution Algorithm (IDEA)}
\label{alg:de_restart}
\begin{algorithmic}[1]
\State Set values for $n_{pop}$, $C_R$, $F$, $iun_{max}$, and $tol_{conv}$, set
$n_{feval}=0$ and $k=1$

 \State Initialize $\mathbf{x}_{i,k}$ and
$\mathbf{v}_{i,k}$ for all $i\in[1,...,n_{pop}]$

\State Create the vector of random values $\mathbf{r}\in U[0, 1]$
and the mask $\mathbf{e}=\mathbf{r}<C_R$

\ForAll{$i\in[1,...,n_{pop}]$}

\State Select three individuals $\mathbf{x}_{i_1}, \mathbf{x}_{i_2},
\mathbf{x}_{i_3}$

\State Create the vector $\mathbf{u}_{i,k} =
\mathbf{e}[(\mathbf{x}_{i_3,k}-\mathbf{x}_{i,k})+F(\mathbf{x}_{i_2,k}-\mathbf{x}_{i_1,k})]$

\State $\mathbf{v}_{i,k+1}=(1-c)\mathbf{v}_{i,k}+\mathbf{u}_{i,k}$
\State Compute $S$ and $\nu$ \State
$\mathbf{x}_{i,k+1}=\mathbf{x}_{i,k}+S\nu \mathbf{v}_{i,k+1}$

\State $n_{feval}=n_{feval}+1$

  \EndFor
\State $k=k+1$

 \State
$\rho_A=\max(\|\mathbf{x}_{i,k}-\mathbf{x}_{j,k}\|)$ for $\forall
\mathbf{x}_{i,k},\mathbf{x}_{j,k}\in P_{sub}\subseteq P_k$

\If{$\rho_A<tol_{conv}\rho_{A,max}$}

\State Run a local optimizer $a_l$ from $\mathbf{x}_{best}$ and let
$\mathbf{x}_l$ be the local minimum found by $a_l$

\If { $f(\mathbf{x}_l)< f(\mathbf{x}_{best})$ }

\State $f_{best} \leftarrow f(\mathbf{x}_l)$ \EndIf
 \If{$f(\mathbf{x}_{best})<f_{min}$}
  \State $f_{min} \leftarrow f(\mathbf{x}_{best})$
  \State $iun=0$
 \Else
  \State $iun=iun+1$
 \EndIf
 \If{$iun\leq iun_{max}$}
  \State Define a bubble $D_l$ such that $\mathbf{x}_{best}\in D_l$ for $\mathbf{x}_{best} \in P_{sub}$ and $\forall P_{sub}\subseteq P_k$

  \State $A_g=A_g+\{\mathbf{x}_{best}\}$ where $\mathbf{x}_{best}=\arg
\min_i f(\mathbf{x}_{i,k})$

  \State Initialize $\mathbf{x}_{i,k}$ and $\mathbf{v}_{i,k}$ for all
$\in[1,...,n_{pop}]$, in the bubble $D_l\subseteq D$
 \Else
 \State Define clusters in the archive and compute the baricenter
 $\mathbf{x}_{c,j}$ of each cluster with $j=1,...,n_{c}$.
  \State  Initialize $\mathbf{x}_{i,k}$ and $\mathbf{v}_{i,k}$ for all $i\in[1,...,n_{pop}]$, in
  $D$ such that $\forall i,j, \|\mathbf{x}_{i,k}-\mathbf{x}_{c,j}\|>\delta_c$
 \EndIf
\EndIf

 \State{\textbf{Termination}} Unless $n_{feval}\geq
n_{fevalmax}$, $goto$ Step~3

\end{algorithmic}
\end{algorithm}

If the population cannot diverge, the eigenvalues cannot have a norm
always $>1$. Furthermore, according to Theorem \ref{thm:fixed_map}
if the function $f$ is strictly quasi-convex in $D$, the population
converges to a single point in $D$, which implies that the map
(\ref{eq:ma_dyn_DE}) is a contraction in $D$ and therefore the
eigenvalues should have a norm on average lower than 1. This can be
illustrated with the following test: Consider a population of 8
individuals and a $D$ enclosing the minimum of a paraboloid with
 the minimum at the origin. For a $C_R$=1.0 and $F$=0.8, we compute, for each step $k$, the distances
of the closest and farthest individuals from the local minimum and
the eigenvalues of the matrix $\mathbf{J}$. Fig. \ref{fig:DE_Diss}
shows the behavior of the eigenvalues and of the distance from the
origin. From the figure, we can see that for all iterations, the
value of the norm of all the eigenvalues is in the interval $[0,1]$
except for one eigenvalue at iteration 12. However, since every expansion is not accepted by the selection function $S$ and for each iteration a number of eigenvalues have modulus lower than 1, the population contracts as
represented in Fig. \ref{fig:DE_Diss}(a).

If multiple minima are contained in $D$, then it can be
experimentally verified that the population contracts to a number of
clusters initially converging to a number of local minima and
eventually to the lowest among all the identified local minima.

The local convergence properties of map (\ref{eq:ma_dyn_DE})
suggest its hybridization with the heuristic
implemented in Monotonic Basin Hopping (MBH).

MBH, first introduced in \cite{WalDoy97} in the context of molecular conformation
problems) is a simple but effective global optimization method. At
each iteration MBH: (i) generates a sample point within a neighborhood of size $2\Delta$ of
the current local minimum (e.g., by a random displacement of each
coordinate of the current local minimum); (ii) starts a local search
from the newly generated sample point; (iii) moves to the newly detected
local minimum only if its function value is better than the function
value at the current local minimum. The initial local minimum is
usually randomly generated within the feasible region. Moreover, if
no improvement is observed for a predetermined number of sample points $n_{samples}$, a
restart mechanism might be activated. The neighborhood of the local minimum $\mathbf{x}_l$ is defined as $[\mathbf{x}_l-\Delta, \mathbf{x}_l+\Delta]^d$. A proper definition of $\Delta$ is essential for the performance of MBH: too small a size would not allow MBH to escape from the current
local minimum, while too large a value would make the search
degenerate in a completely random one. The local search performed at
each iteration can be viewed as  a {\em dynamical system} where the evolution of the systems at each iteration is controlled by some {\em map}. Under suitable assumptions, the systems converge
to a fixed point. For instance, if $f$ is convex and $C^2$ in a
small enough domain containing a local minimum which satisfies some
regularity conditions, Newton's map converges quadratically to a
single fixed point (the local minimum) within the domain. This
observation leads to the above mentioned hybridization of DE with
MBH: the dynamical system corresponding to a local search is
replaced by the one corresponding to DE. More precisely, if map
(\ref{eq:ma_dyn_DE}), either for a cluster $P_{sub}\subseteq P_k$ or
for the entire population $P_k$, contracts, we can define a bubble
$D_l\subseteq D$, around the best point within the cluster
$\mathbf{x}_{best}$, and re-initialize a subpopulation $P_{sub}$ in
$D_l$. Such operation is performed as soon as the maximum distance
$\rho_A=\max(\|\mathbf{x}_i-\mathbf{x}_j\|)$ among the elements in
the cluster collapses below a value $tol_{conv}\rho_{A,max}$, where
$\rho_{A,max}$ is the maximum $\rho_{A}$ recorded during the
convergence of the map (\ref{eq:ma_dyn_DE_matrix}).


In order to speed up convergence, the best solution
$\mathbf{x}_{best}$ of the cluster is refined through a local search
started at it, leading to a local minimum $\mathbf{x}_l$, which is
saved in an archive $A_g$. The bubble is defined, similarly to the neighborhood of MBH, as
$D_l=[\mathbf{x}_{l}-\Delta \; \mathbf{x}_{l}+\Delta]^d$. The
overall process leads to Algorithm \ref{alg:de_restart}. Note that
the contraction of the population given, for example, by the metric
$\rho_A$, is a stopping criterion that does not depend explicitly on
the value of the objective function but on the contractive
properties of the map in Eq. \eqref{eq:ma_dyn_DE}. Some remarks
follow.
\begin{itemize}
\item[1.] Convergence of DE to a single point can not always be guaranteed, and consequently, we can not always
guarantee that the contraction condition $\rho_A\leq tol_{conv}\rho_{A,max}$ will be satisfied at some iteration.
Therefore, in order to take into account this possibility, we need to introduce a further stopping criterion for DE,
such as a maximum number of iterations. We point out, however, that such alternative stopping criterion has never become active in our experiments.
\item[2.] Even when DE converges to a single point, this is not guaranteed to be a local minimum. For this reason
we always refine the best observed solution through a local search.
\item[3.] If the search space is characterized by a single funnel structure
\cite{Loca05}, the restart of the population in the bubble allows
the algorithm to move towards the global optimum by progressively
jumping from one minimum to a better one. On the other hand, if
multiple funnels or multiple isolated minima exist, a simple restart
of the population inside a bubble might not be sufficient to avoid
stagnation. A way to overcome this problem is to use global
re-sampling: when the value of the best solution does not change for
a predefined number of iterations $iun_{max}$, the population is
restarted. The restart procedure collects the solutions in the
archive into $n_c$ clusters with baricenter $x_{c,j}$ for
$j=1,...,n_c$, then each agent $\mathbf{x}_i$ of the new population
is generated so that $\|\mathbf{x}_i-\mathbf{x}_{c,j}\|>\delta_c$.
Note that a restart mechanism for DE was previously
proposed also by Peng et al. \cite{peng:09} in a variant of the adaptive
Differential Evolution JADE. However, in the work of
Peng et al. the restart criterion and restart strategy are
substantially different from the ones proposed here. For example,
although we record the local minima in an external archive, we do
not prevent the algorithm from searching in the surroundings of the
recorded minima. On the contrary, we combine a local restart in a
bubble surrounding the final point returned by DE, according to the heuristics of
MBH, with a more global restart sampling outside the bubbles.
A complete review of the existing variants of Differential Evolution can found in the work of Neri et al. \cite{neri10}.
Other restart mechanisms have recently been
adopted to improve other evolutionary algorithms such as G-CMAES
\cite{auger:05} or hybrid methods \cite{sentinella:07}, but also in these cases
the restart criterion and restart strategy are
substantially different from the ones proposed here.
\item[4.] A question, mainly of theoretical interest, is about convergence.
General results stating conditions under which convergence to the
global minimum (with probability 1) is guaranteed can be found,
e.g., in \cite{pinter:84} and \cite{rudolph:96}. Such conditions are
not fulfilled by standard DE, while IDEA fulfills them if parameter
$\delta_c$ employed in the restart mechanism is "small enough" (if
too large, some portions of the feasible region, possibly including
the global minimum, might remain unexplored). We point out
that, although in practice we are not interested in the behavior of
an algorithm over an infinite time horizon, global convergence
justifies re-running the search process with an increased number of
function evaluations should the results be unsatisfactory. As we will
show in the remainder of this paper, IDEA produces performance
steadily increasing with the number of function evaluations without
the need to change its settings.
\end{itemize}

\section{Trajectory Model and Problem Formulation}
The modified differential evolution algorithm derived in Section
\ref{sec:agent_dyna} is applied to the solution of four real world
cases. The four cases are all multigravity assist (MGA) trajectory design
problems, three of them with deep space manoeuvres (DSM) and one
with no DSM's. In this section we describe the trajectory model and
we formulate the global optimization problem that will be tackled
through Algorithm \ref{alg:de_restart}.

\subsection{Trajectory Model with no DSM's}\label{sec:trj_nodsm}
Multi-gravity assist transfers with no deep space maneuvers can be
modeled with a sequence of conic arcs connecting a number of
planets. The first one is the departure planet, the last one is the
destination planet and at all the intermediate ones the spacecraft
performs a gravity assist maneuver. Given $N_P$ planets $P_i$ with
$i=1,...,N_P-1$, each conic arc is computed as the solution of a
Lambert's problem \cite{battin:99} given the departure time from
planet $P_i$ and the arrival time at planet $P_{i+1}$. The solution
of the Lambert's problems yields the required incoming and outgoing
velocities at a swing-by planet $v_{in}$ and $v_{rout}$ (see Fig.
\ref{fig:whole_trj}). The swing-by is modeled through a linked-conic
approximation with powered maneuvers\cite{becerra:04}, i.e., the
mismatch between the required outgoing velocity $v_{rout}$ and
 the achievable outgoing velocity $v_{aout}$ is compensated through a
$\Delta v$ maneuver at the pericenter of the gravity assist
hyperbola. The whole trajectory is completely defined by the
departure time $t_0$ and the transfer time for each leg $T_i$, with
$i=1,...,N_P-1$.


\begin{figure}[tb]
\begin{center}
\vskip -5mm
 \includegraphics[width=4.5cm]{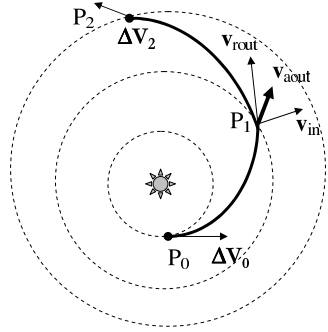}
\vskip -5mm \caption{Trajectory model with no DSM's}
\label{fig:whole_trj}
\end{center}
\end{figure}

The normalized radius of the pericenter $r_{p,i}$ of each swing-by
hyperbola is derived a posteriori once each powered swing-by
manoeuvre is computed. Thus, a constraint on each pericenter radius
has to be introduced during the search for an optimal solution. In
order to take into account this constraint, the objective function
is augmented with the weighted violation of the constraints:
\begin{equation}\label{eq:evvejs_fobj}
    f(\mathbf{x})=\Delta v_0+\sum_{i=1}^{N_p-2}\Delta v_i + \Delta
v_f+\sum_{i=1}^{N_p-2}w_i(r_{p,i}-r_{pmin,i})^2
\end{equation}
for a solution vector:
\begin{equation}\label{eq:generalMGA_x}
    \mathbf{x}=[t_0,T_1,T_2,...,T_{N_P-1}]^T
\end{equation}

\subsection{Trajectory Model with DSM's}\label{sec:trj_dsm}
A general MGA-DSM trajectory can be modeled through a sequence of
$N_P-1$ legs connecting $N_P$ celestial bodies (Fig.
\ref{fig:mga_trj_model})\cite{vasile:06}. In particular if all
celestial bodies are planets, each leg begins and ends with an
encounter with a planet. Each leg $i$ is made of two conic arcs: the
first, propagated analytically forward in time, ends where the
second solution of a Lambert's problem begins. The two arcs have a
discontinuity in the absolute heliocentric velocity at their
matching point $M$. Each DSM is computed as the vector difference
between the velocities along the two conic arcs at the matching
point. Given the transfer time $T_i$ and the variable $\alpha_i\in
[0,1]$ relative to each leg $i$, the matching point is at time
$t_{DSM,i}=t_{f,i-1}+\alpha_iT_i$, where $t_{f,i-1}$ is the final
time of the leg $i-1$. The relative velocity vector $\mathbf{v}_0$
at the departure planet can be a design parameter and is expressed
as:
\begin{equation}\label{eq:mga_v0}
    \mathbf{v}_0=v_0[\sin\delta \cos\theta, \sin\delta \sin\theta,
\cos\delta]^T
\end{equation}
with the angles $\delta$  and $\theta$ respectively representing the
declination and the right ascension with respect to a local
reference frame with the $x$ axis aligned with the velocity vector
of the planet, the $z$ axis normal to orbital plane of the planet
and the $y$ axis completing the coordinate frame. This choice allows for an
easy constraint on the escape velocity and asymptote direction
while adding the possibility of having a deep space maneuver in the
first arc after the launch. This is often the case when the escape
velocity must be fixed due to the launcher capability or to the
requirement of a resonant swing-by of the Earth (Earth-Earth
transfers). In order to have a uniform distribution of random points
on the surface of the sphere defining all the possible launch
directions, the following transformation has been applied:
\begin{equation}\label{eq:mga_sphereangle}
  \bar{\theta}=\frac{\theta}{2\pi} \qquad \bar{\delta}=\frac{\cos(\delta +\pi/2)+1}{2}
\end{equation}
It results that the sphere surface is uniformly sampled when a
uniform distribution of points for $\bar{\theta},\bar{\delta}\in
[0,1]$ is chosen. Once the heliocentric velocity at the beginning of
leg $i$, which can be the result of a swing-by maneuver or the
asymptotic velocity after launch, is computed, the trajectory is
analytically propagated until time $t_{DSM,i}$. The second arc of
leg $i$ is then solved through a Lambert's algorithm, from $M_i$,
the Cartesian position of the deep space maneuver, to $P_i$, the
position of the target planet of phase $i$, for a time of flight
$(1-\alpha_i)T_i$. Two subsequent legs are then joined together with
a gravity assist manoeuvre. The effect of the gravity of a planet is
to instantaneously change the velocity vector of the spacecraft.

The relative incoming velocity vector and the outgoing velocity vector
at the planet swing-by have the same modulus but different
directions; therefore the heliocentric outgoing velocity results to
be different from the heliocentric incoming one. In the linked conic
model, the spacecraft is assumed to follow a hyperbolic trajectory
with respect to the swing-by planet. The angular difference between
the incoming relative velocity $\tilde{\mathbf{v}}_i$  and the
outgoing one $\tilde{\mathbf{v}}_o$ depends on the modulus of the
incoming velocity and on the pericenter radius $r_i$. Both the
relative incoming and outgoing velocities belong to the plane of the
hyperbola. However, in the linked-conic approximation, the maneuver
is assumed to occur at the planet, where the planet is a point mass
coinciding with its center of mass. Therefore, given the incoming
velocity vector, one angle is required to define the attitude of the
plane of the hyperbola $\Pi$. Although there are different possible choices
for the attitude angle $\gamma$, the one proposed in Ref. 7 has been
adopted (see Fig. \ref{fig:ga_model}), where $\gamma$ is the angle between the
vector $\mathbf{n}_{\Pi}$, normal to the hyperbola plane $\Pi$, and
the reference vector $\mathbf{n}_r$, normal to the plane
containing the incoming relative velocity and the velocity of the
planet $\mathbf{v}_{P}$.
\begin{figure}[tb]
\begin{center}
\vskip -5mm
 \includegraphics[width=7.5cm]{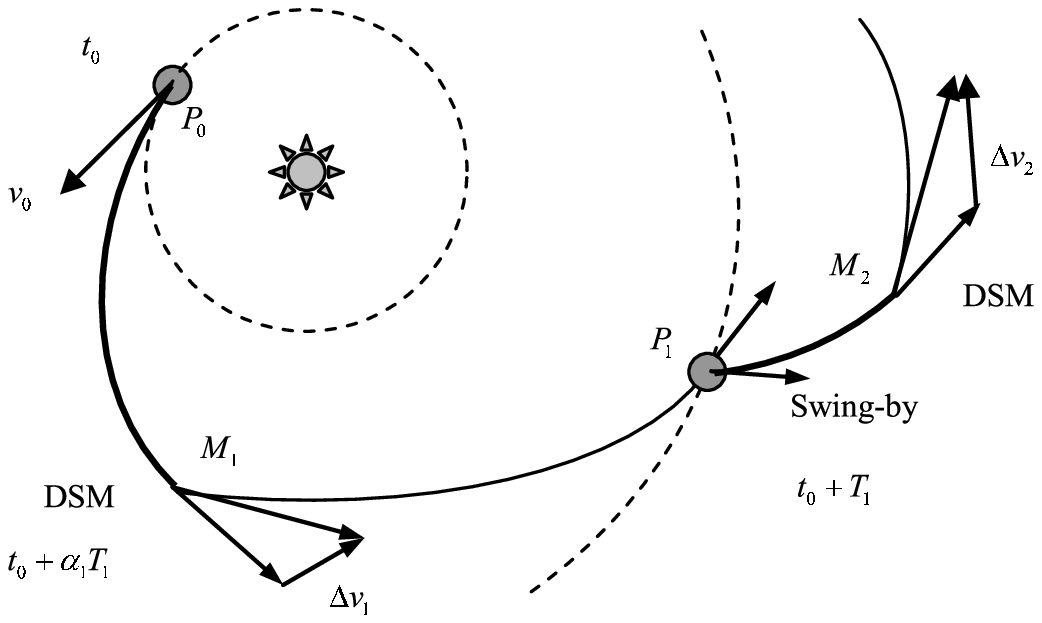}
\vskip -5mm \caption{Schematic representation of a multiple gravity
assist trajectory} \label{fig:mga_trj_model}
\end{center}
\end{figure}

\begin{figure}[tb]
\begin{center}
\vskip -5mm
 \includegraphics[width=5.5cm]{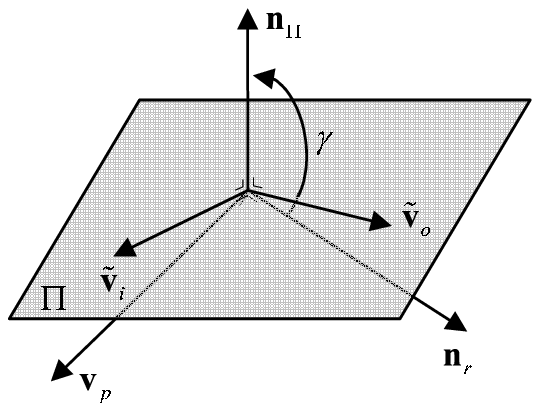}
\vskip -5mm \caption{Schematic representation of a multiple gravity
assist trajectory} \label{fig:ga_model}
\end{center}
\end{figure}

Given the number of legs of the trajectory $N_L=N_P-1$, the complete
solution vector for this model is:
\begin{align}\label{eq:mga_x}
 \mathbf{x}=&[v_0,\bar{\theta},\bar{\delta},t_0,\alpha_1,T_1,
    \gamma_1,r_{p,1},\alpha_2,T_2,..., \\
    &\gamma_i,r_{p,i},T_{i-1},\alpha_{i-1},...,\gamma_{N_L-1},r_{p,N_L-1},\alpha_{N_L},T_{N_L}] \nonumber
\end{align}
where $t_0$ is the departure date. Now, the design of a
multi-gravity assist transfer can be transcribed into a general
nonlinear programming problem, with simple box constraints, of the
form:
\begin{equation}\label{eq:min_fobj}
    \min_{\mathbf{x}\in D}f(\mathbf{x})
\end{equation}

One of the appealing aspects of this formulation is its solvability
through a general global search method for box constrained problems.
Depending on the kind of problem under study, the objective function
can be defined in different ways. Here we choose to focus on
minimizing the total $\Delta v$ of the mission, therefore the
objective function $f(\mathbf{x})$ is:
\begin{equation}\label{eq:mga_fobj}
    f(\mathbf{x})=v_0+\sum_{i=1}^{N_p}\Delta v_i + \Delta v_f
\end{equation}
where $\Delta v_i$ is the velocity change due to the DSM in the
$i$-th leg, and $\Delta v_f$ is the maneuver needed to inject the
spacecraft into the final orbit.

\section{Test Problems}
We consider a benchmark made of four different test-cases:  two
versions of the MGA transfer from the Earth to Saturn of the
Cassini-Huygens mission, a multi-gravity assist transfer to the
comet 67P/Churyumov-Gerasimenko (similar to the Rosetta mission),
and a multi-gravity assist rendezvous transfer with mid-course
manoeuvres to Mercury (similar to the Messenger mission). Algorithm
\ref{alg:de_restart}, called IDEA, is applied to the solution of the
four cases and compared to standard Differential Evolution
\cite{kenneth:05} and Monotonic Basin Hopping \cite{WalDoy97,
Addis-et-al}.

\subsection{Cassini with no DSM's}

The first test case is a multi gravity assist trajectory from the
Earth to Saturn following the sequence
Earth-Venus-Venus-Earth-Jupiter-Saturn (EVVEJS). There are six
planets and the transfer is modeled as in Section
\ref{sec:trj_nodsm}, thus the solution vector is:
\begin{equation}\label{eq:evvejs_x}
    \mathbf{x}=[t_0,T_1,T_2,T_3,T_4,T_5]^T
\end{equation}

The final $\Delta v_f$ is the $\Delta v$ needed to inject the
spacecraft into an ideal operative orbit around Saturn with a
pericenter radius of 108950 km and an eccentricity of 0.98. The
weighting functions $w_i$ are defined as follows:
\begin{align}\label{eq:evvejs_w}
  w_i &= 0.005[1-\mathrm{sign}(r_{p,i}-r_{pmin,i})], \quad i=1,...,3\\
  w_4&=0.0005[1-\mathrm{sign}(r_{p,4}-r_{pmin,4})]  \nonumber
\end{align}
with the minimum normalized pericenter radii $r_{pmin,1}=1.0496$,
$r_{pmin,2}=1.0496$, $r_{pmin,3}=1.0627$ and $r_{pmin,4}=9.3925$.
For this case the dimensionality of the problem is six, with the
search space $D$ defined by the following intervals: $t_0\in [-1000,
0]$MJD2000, $T_1\in [30, 400]$d, $T_2\in [100, 470]$d, $T_3\in [30,
400]$d, $T_4\in [400, 2000]$d, $T_5\in [1000, 6000]$d. The best
known solution is $f_{best}=4.9312$ km/s, with
$\mathbf{x}_{best}$=[--789.75443770458, 158.301628961437, 449.385882183958, 54.7050296906556, 1024.5997453164, 4552.72068790619]$^T$.

\subsection{Cassini with DSM's}

The second test case is again a multi gravity assist trajectory from
the Earth to Saturn following the sequence
Earth-Venus-Venus-Earth-Jupiter-Saturn (EVVEJS), but a deep space
manoeuvre is allowed along the transfer arc from one planet to the
other according to the model presented in Section \ref{sec:trj_dsm}.
Although from a trajectory design point of view, this problem is
similar to the first test case, the model is substantially different
and therefore represents a different problem from a global
optimization point of view. Since the transcription of the same
problem into different mathematical models can affect the search for
the global optimum, it is interesting to analyze the behavior of the
same set of global optimization algorithms applied to two different
transcriptions of the same trajectory design problem.


Here, $\Delta v_f$ is defined as the modulus of the vector difference between the
velocity of Saturn at arrival and the velocity of the spacecraft at
the same time. For this case the dimensionality of the problem is
22, with the search space $D$ is defined by the following intervals:
$t_0\in [-1000, 0]$MJD2000,  $v_0\in[3, 5]$km/s, $\bar{\theta}\in[0,
1]$, $\bar{\delta}\in[0, 1]$, $T_1\in [100, 400]$d, $T_2\in [100,
500]$d, $T_3\in [30, 300]$d, $T_4\in [400, 1600]$d, $T_5\in [800,
2200]$d, $\alpha_1\in [0.01, 0.9]$, $\alpha_2\in [0.01, 0.9]$,
$\alpha_3\in [0.01, 0.9]$, $\alpha_4\in [0.01, 0.9]$, $\alpha_5\in
[0.01, 0.9]$, $r_{p,1}\in[1.05, 6]$, $r_{p,2}\in[1.05, 6]$,
$r_{p,3}\in[1.15, 6.5]$, $r_{p,4}\in[1.7, 291]$, $\gamma_1\in[0,
2\pi]$, $\gamma_2\in[0, 2\pi]$, $\gamma_3\in[0, 2\pi]$,
$\gamma_4\in[0, 2\pi]$.
The best known solution is
$f_{best}=8.3889$ km/s, $\mathbf{x}_{best}$=[--780.917853635368, 3.27536879103551, 0.782513100225235, 0.378682006044345, 169.131920055057, 424.13242396494, 53.296452710059, 2199.98648654574, 0.795774035295027, 0.530055267286, 0.126002760289258,        0.0105947672634,  0.0381505843013, 1.35556902792788, 1.05001740672886, 1.30699201995999, 71.3749247783128,          3.15842153037544, 3.53046280721895, 3.12561791754698, 3.08422162979462]$^T$.

\subsection{Rosetta Mission}
The third test case is a multi gravity assist trajectory from the
Earth to the comet 67P/Churyumov-Gerasimenko following the gravity
assist sequence that was planned for the spacecraft Rosetta:
Earth-Earth-Mars-Earth-Earth-Comet. The trajectory model is the one
described in Section \ref{sec:trj_dsm} but the objective function does not include $v_0$.

For this case the dimensionality of the problem is
22, with the search space $D$ is defined by the following intervals:
$t_0\in [1460, 1825]$MJD2000,  $v_0\in[3, 5]$km/s, $\bar{\theta}\in[0,
1]$, $\bar{\delta}\in[0, 1]$, $T_1\in [300, 500]$d, $T_2\in [150,
800]$d, $T_3\in [150, 800]$d, $T_4\in [300, 800]$d, $T_5\in [700,
1850]$d, $\alpha_1\in [0.01, 0.9]$, $\alpha_2\in [0.01, 0.9]$,
$\alpha_3\in [0.01, 0.9]$, $\alpha_4\in [0.01, 0.9]$, $\alpha_5\in
[0.01, 0.9]$, $r_{p,1}\in[1.05, 9]$, $r_{p,2}\in[1.05, 9]$,
$r_{p,3}\in[1.05, 9]$, $r_{p,4}\in[1.05, 9]$, $\gamma_1\in[0,
2\pi]$, $\gamma_2\in[-\pi, \pi]$, $\gamma_3\in[0, 2\pi]$,
$\gamma_4\in[0, 2\pi]$.

The best known solution is $f_{best}$=1.34229 km/s, with solution vector
$x_{best}$=[1542.65536672006, 4.48068107888312, 0.935220667497966, 0.9909562486258, 365.24235847396, 707.540858648698, 257.417859715383, 730.483434305258, 1850, 0.310501108489873, 0.809061227121068, 0.0124756484551758, 0.0466967002704, 0.43701236871638, 1.8286351998512, 1.05, 2.80973511169638, 1.18798981835459, 2.61660601734377, --0.215250274241349, 3.57950314115394, 3.46467471264343]$^T$.

\subsection{Messenger Mission}
The fourth problem is a multi-gravity assist trajectory from the
Earth to planet Mercury following the sequence of planetary
encounters of the first part of the Messenger mission. As for the
previous test case, the trajectory model is the one described in
Section \ref{sec:trj_dsm}.

For this case the dimensionality of the problem is
18, with the search space $D$ is defined by the following intervals:
$t_0\in [1000, 4000]$MJD2000,  $v_0\in[1, 5]$km/s, $\bar{\theta}\in[0,
1]$, $\bar{\delta}\in[0, 1]$, $T_1\in [200, 400]$d, $T_2\in [30,
400]$d, $T_3\in [30, 400]$d, $T_4\in [30, 400]$d, $\alpha_1\in [0.01, 0.99]$, $\alpha_2\in [0.01, 0.99]$,
$\alpha_3\in [0.01, 0.99]$, $\alpha_4\in [0.01, 0.99]$, $r_{p,1}\in[1.1, 6]$, $r_{p,2}\in[1.1, 6]$,
$r_{p,3}\in[1.1, 6]$, $\gamma_1\in[-\pi,
\pi]$, $\gamma_2\in[-\pi, \pi]$, $\gamma_3\in[-\pi, \pi]$.

The best known solution is $f_{best}$ =
8.631 km/s, with solution vector $\mathbf{x}_{best}
$=[1171.14619813253, 1.41951376601752, 0.628043728560056, 0.500000255697689, 399.999999999969, 178.921469111868, 299.279691870106, 180.689114497891, 0.236414009949924, 0.0674215615945254, 0.832992171208578, 0.312514378885353, 1.7435422021558, 3.03087330660319, 1.10000000000119, 0.219820823285448, 0.477475660779879, 0.225898117795826]$^T$.

\vspace{0.5cm}
 Note that, the search space for each one of the
trajectory models is normalized so that $D$ is a unit hypercube with
each component of the solution vector belonging to the interval
$[0,1]$. Furthermore, for all cases, solution algorithms were run
for a progressively increasing number of function evaluations from
100000 to 1.25 million.

\section{Testing Procedure}\label{sec:testing}
The modified DE algorithm will be compared against standard DE
and MBH following a rigorous testing procedure. A detailed
description of the testing procedure can be found in
\cite{cit:vasileAAS2008} and it is here summarized in Algorithm
\ref{alg:pro_conv} for a generic solution algorithm $A$ and a
generic problem $p$. Here $\bar{\mathbf{x}}(A,i)$ denotes the best
point observed during the $i$-th run of algorithm~$A$.

\begin{algorithm}[!h]
\caption {Testing Procedure} \label{alg:pro_conv}
\begin{algorithmic}[1]
\State Set to $N$ the max number of function evaluations for $A$
\State Apply $A$ to $p$ for $n$ times and set $j_s=0$
\ForAll{$i\in[1,...,n]$}\\
  Compute $\delta_f(\bar{\mathbf{x}}(A,i))=\mid f_{global}-f(\bar{\mathbf{x}}(A,i))\mid;$ and $\delta_x(\bar{\mathbf{x}}(A,i))= \|\mathbf{x}_{global}-\bar{\mathbf{x}}(A,i)
\|$
    \If{ ($\delta_f(\bar{\mathbf{x}}(A,i))<tol_f$) $\wedge$ ($\delta_x(\bar{\mathbf{x}}(A,i))<tol_x$)}
    $j_s=j_s+1$
    \EndIf
  \EndFor
\end{algorithmic}
\end{algorithm}

The index of performance $j_s$ is the number of successes of the
algorithm $A$. In the following we use the $f_{best}$ values
reported above in place of $f_{global}$ and we consider only
$\delta_f(\bar{\mathbf{x}}(A,i))$ and not
$\delta_x(\bar{\mathbf{x}}(A,i))$.

The success rate, $p_s=j_s/n$, will be used for the comparative assessment
of the algorithm performance instead of the commonly used best
value, mean and variance. Indeed, the distribution of the
function values is not Gaussian. Therefore, the average value can be
far away from the results returned with a higher frequency from a
given algorithm. In the same way, the variance is not a good
indicator of the quality of the algorithm because a high variance
together with a high mean value can correspond to the case in which
50\% of the results are close to the global optimum with the other
50\% far from it. Finally, statistical tests, such as the t-test,
that assume a Gaussian distribution of the sample can not be applied
to correctly predict the behavior of an algorithm. Instead, the
success rate gives an immediate and unique indication of the
algorithm effectiveness, and, moreover, it can be always represented
with a binomial probability density function (pdf), independent of
the number of function evaluations, the problem and the type of
optimization algorithm.

 A key point is setting properly the value of
$n$ to have a reliable estimate of the success probability of an
algorithm, or success rate $p_s=j_s/n$. Since the success is binomial (assumes values that are either 0 or
1) we can set a priori the value of $n$ to get the required
confidence on the value of $p_s$. A commonly adopted starting point
for sizing the sample of a binomial distribution is to assume a
normal approximation for the sample proportion $p_s$ of successes
(i.e. $p_s \thicksim N\lbrace\theta_p, \theta_p(1-\theta_p)/n
\rbrace$, where $\theta_p$ is the unknown true proportion of
successes) and the requirement that $ Pr [|p_s-\theta_p| \leq
d_{err}|\theta_p ]$ is at least equal to $ 1-\alpha_p$
\cite{adcock:97}. This leads to the expression $n \geq
\theta_p(1-\theta_p) \chi^{2}_{(1),\alpha_p}/ d_{err}^2$ that can be
approximated conservatively with
\begin{equation}\label{eq:nruns}
n \geq 0.25 \chi^{2}_{(1),\alpha_p}/ d_{err}^2
\end{equation}
valid for $\theta_p=0.5$. According to Eq. (\ref{eq:nruns}), an error $\leq
0.05$ ($d_{err}$ = 0.05) with a 95\% of confidence ($\alpha_p$ =
0.05) would require at least $n = 175$. If $n$ is extended to 1000, the error reduces to 0.020857. In the following, we will use $n=200$ and $N=1.25e6$ for tuning the algorithms, and $n=1000$ with variable $N$ to compare their performance. In fact, a reduced error is required to discriminate between the performance of two algorithms at low $N$.

The values of $tol_f$ for the four test cases are: $tol_f=0.0688$ km/s for Cassini with no DSM's, $tol_f=0.1111$
km/s for Cassini with DSM's, $tol_f=0.05778$ km/s for Rosetta and
$tol_f=0.05$ km/s for Messenger. The choice of these thresholds was dictated by the need to discriminate among different minima. At the same time they represent a reasonable margin on the total $\Delta v$. In fact, during standard preliminary designs, a margin between 3\% and 5\% is typically added for contingencies, while here all the selected thresholds are below 5\% of the value of the objective function. In other words, all the minima within the thresholds would be indistinguishable from a mission design point of view.

\subsection{Parameter Tuning}
Prior to running the tests on all the four cases, the key parameters
for DE and MBH were tuned to get the best performance. We used the
Rosetta case as a tuning example. The tuning of DE and MBH is used
also to tune IDEA. The tuning of MBH was relatively fast as there are only
two parameters: the size of the neighborhood and the number of samples before global restart.

Figure~\ref{fig:MBH_Rosetta_perfo} shows the performance of MBH on the Rosetta case for different values of $\Delta$ and $n_{samples}$. It can be noted that the performance tends to increase for a number of samples that tend to infinity. On the other hand, there is a peak of performance around $n_{samples}=30$. In the remainder of this paper, we will call MBH-GR the version of MBH implementing a global restart after 30 unsuccessful samples and we call MBH, the version with $n_{samples}=\infty$. The optimum $\Delta$ seems to be 0.1, therefore this value was used for all the test cases.

Note that the general trend of the performance of MBH does not change for the other cases and therefore the settings seem to be of general validity for this benchmark of test problems.

\begin{figure}[tb]
\begin{center}
 \includegraphics[width=8cm]{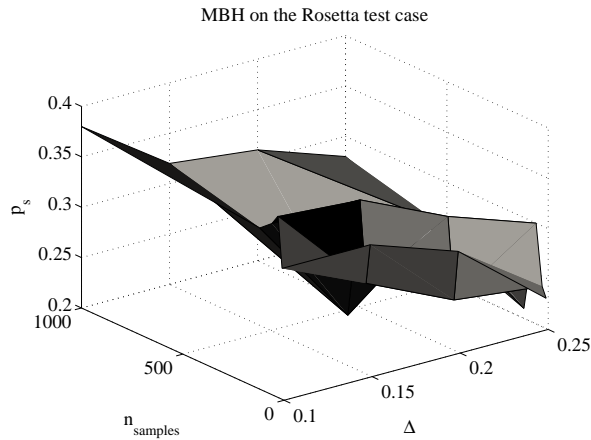}
\vskip -5mm
 \caption{Performance of MBH on the Rosetta test case} \label{fig:MBH_Rosetta_perfo}
\end{center}
\end{figure}

For the tuning of DE, we instead considered a grid of values for $F$,
$C_R$ and population size, for different strategies. Figs. \ref{fig:DE_Rosetta_perfo6} to \ref{fig:DE_Rosetta_perfo7pop} show the most significant
results. From these figures one can deduce that strategy \emph{DE/best/1/bin} with $F$
in the range $[0.8,\; 1]$ and a population below 200 individuals
would be a good choice. Alternatively strategy \emph{DE/rand/1/bin} can be used with a
smaller population and an $F$ in the range $[0.6, \;0.7]$. In both
cases $C_R$ is better at around 0.8 as the problem is not
separable and the components of the solution vector are
interdependent.

The combination of strategy and values of $F$, $n_{pop}$ controls
the speed of local convergence to a fixed point of the algorithm. In
the following tests, therefore, we decided to use strategy \emph{DE/best/1/bin} with two
sets of populations, $[5\,d, 10\,d]$, where $d$ is the
dimensionality of the problem, with single values of step-size and
crossover probability $F=0.75$ and $C_R=0.8$ respectively. The two
settings will be denoted with DE5c, DE10c.
The trends in Figs. \ref{fig:DE_Rosetta_perfo6} to \ref{fig:DE_Rosetta_perfo7pop} can be registered also for the other two test cases, although the performance of DE is much poorer than for Rosetta, therefore it can be argued that the settings have general validity. These settings are also in line with the theory developed by Zaharie in \cite{zaharie:02}.

 The settings of IDEA were derived from the individual tuning of
DE and MBH. In particular, we took a value $C_R=0.9$, as the
variables are not decoupled, a convergence strategy for the choice
of the indexes in Eq. (\ref{eq:de}) and a value $F=0.9$ together
with a small population to have a fast convergence but without
losing exploration capabilities. We used an $n_{pop}=20$ for
Cassini with no DSM's and for Messenger, while an $n_{pop}=40$ for
Rosetta and Cassini with DSM's. For all the test cases
$\delta_c=0.1$, $tol_{conv}=0.25$, while $\Delta=0.2$ for all the
cases except for Messenger for which we used $0.25$ instead. The parameter controlling the maximum
number of local restarts is $iun_{max}=6$ for Messenger,
$iun_{max}=2$ for Rosetta, $iun_{max}=+\infty$ for Cassini with and
without DSM's as no guided restart is applied.

\begin{figure}[tb]
\begin{center}
 \includegraphics[width=8cm]{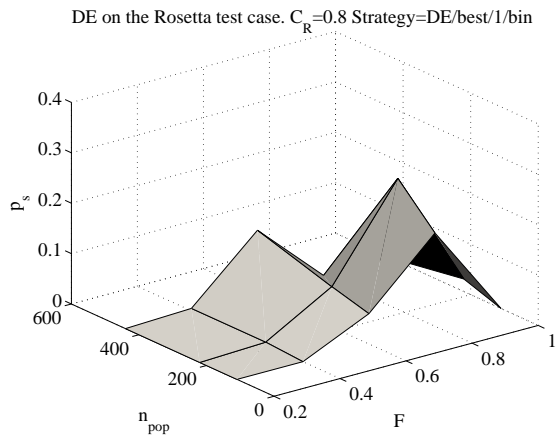}
\vskip -5mm
 \caption{Performance of Differential Evolution on the Rosetta test case, with strategy \emph{DE/best/1/bin} and $C_R$=0.8} \label{fig:DE_Rosetta_perfo6}
\end{center}
\end{figure}

\begin{figure}[tb]
\begin{center}
 \includegraphics[width=8cm]{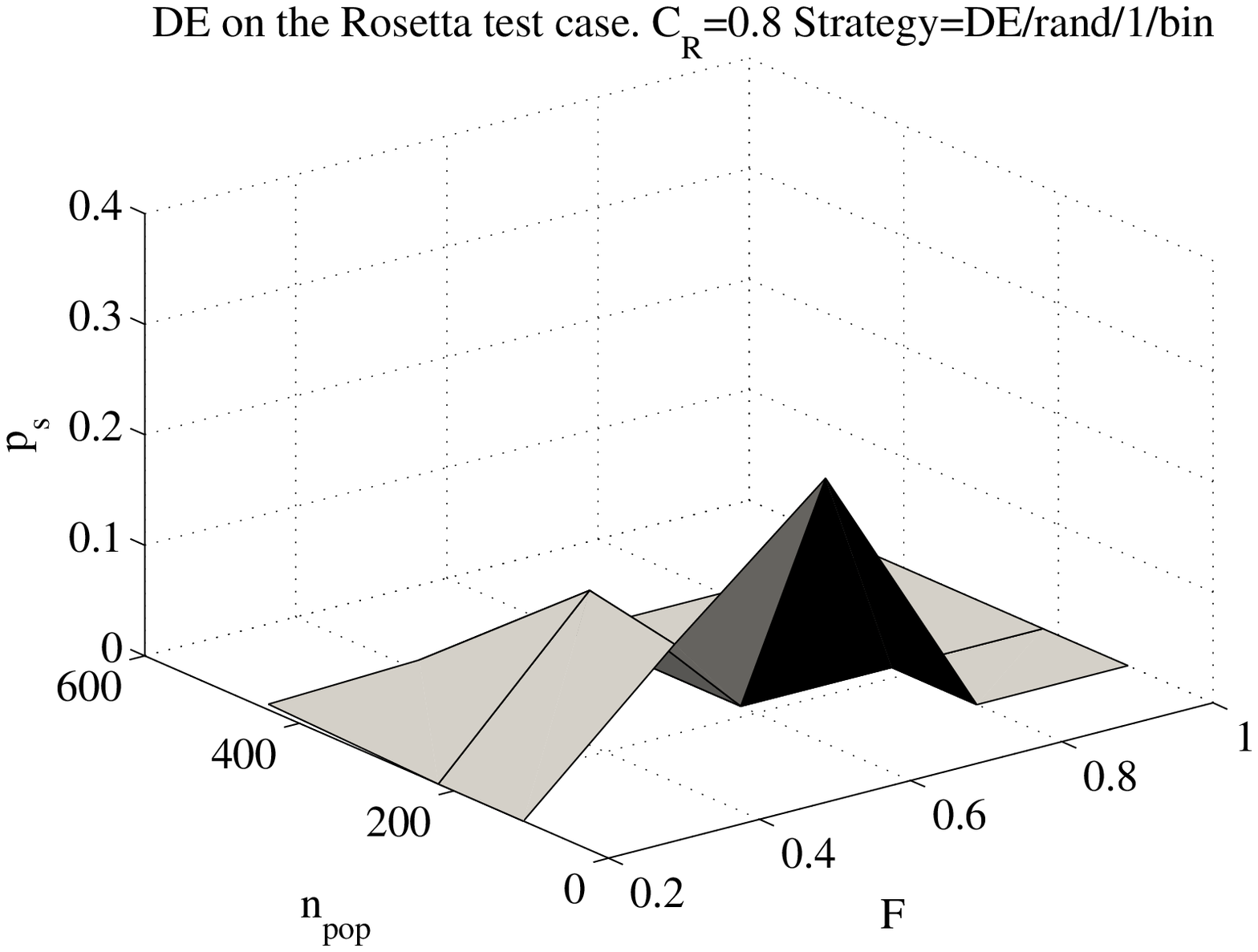}
\vskip -5mm
 \caption{Performance of Differential Evolution on the Rosetta test case, with strategy \emph{DE/rand/1/bin} and $C_R$=0.8} \label{fig:DE_Rosetta_perfo7}
\end{center}
\end{figure}

\begin{figure}[tb]
\begin{center}
 \includegraphics[width=8cm]{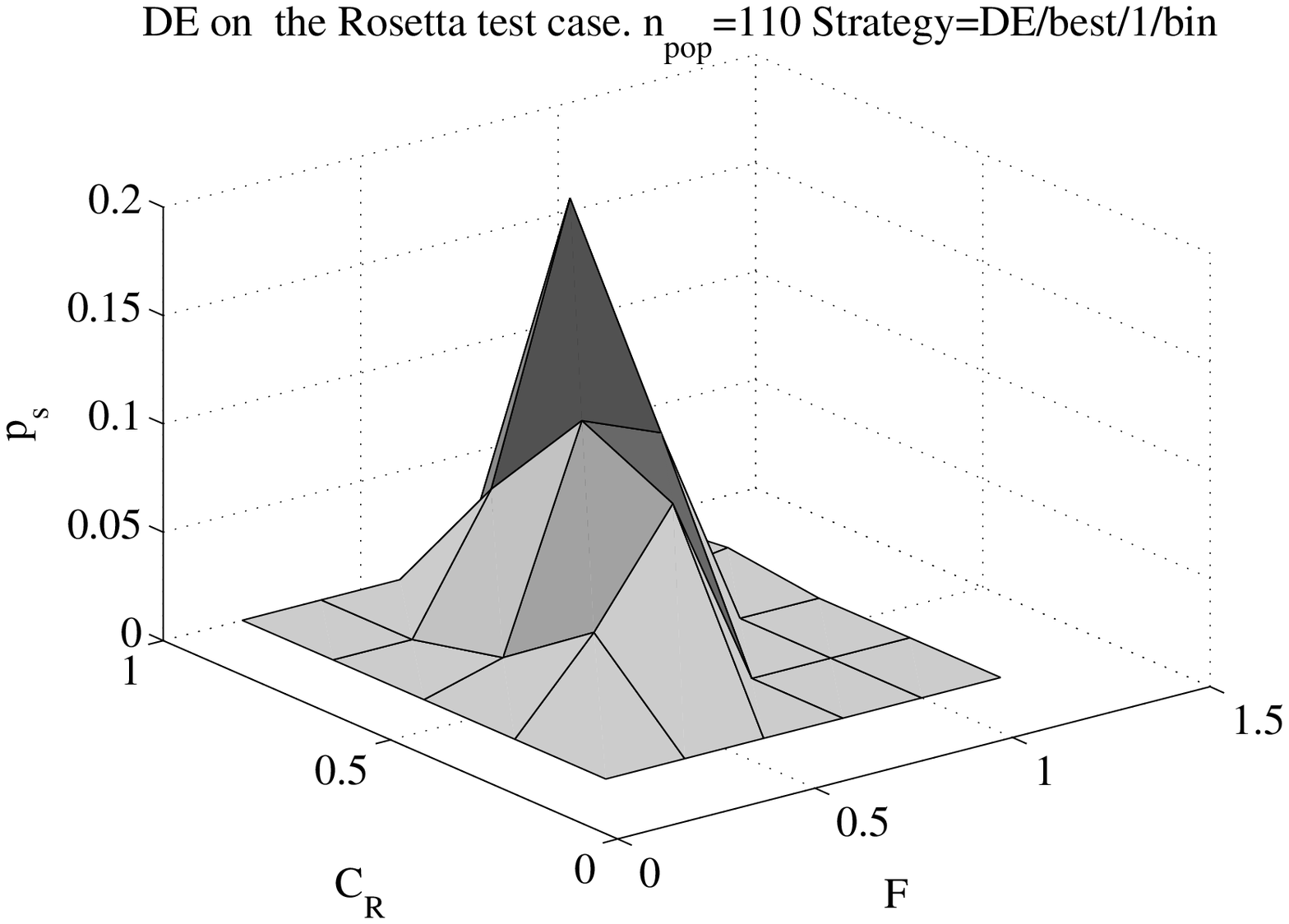} 
 \vskip -5mm
 \caption{Performance of Differential Evolution on the Rosetta test case, with strategy \emph{DE/best/1/bin} and $n_{pop}$=110} \label{fig:DE_Rosetta_perfo6pop}
\end{center}
\end{figure}

\begin{figure}[tb]
\begin{center}
 \includegraphics[width=8cm]{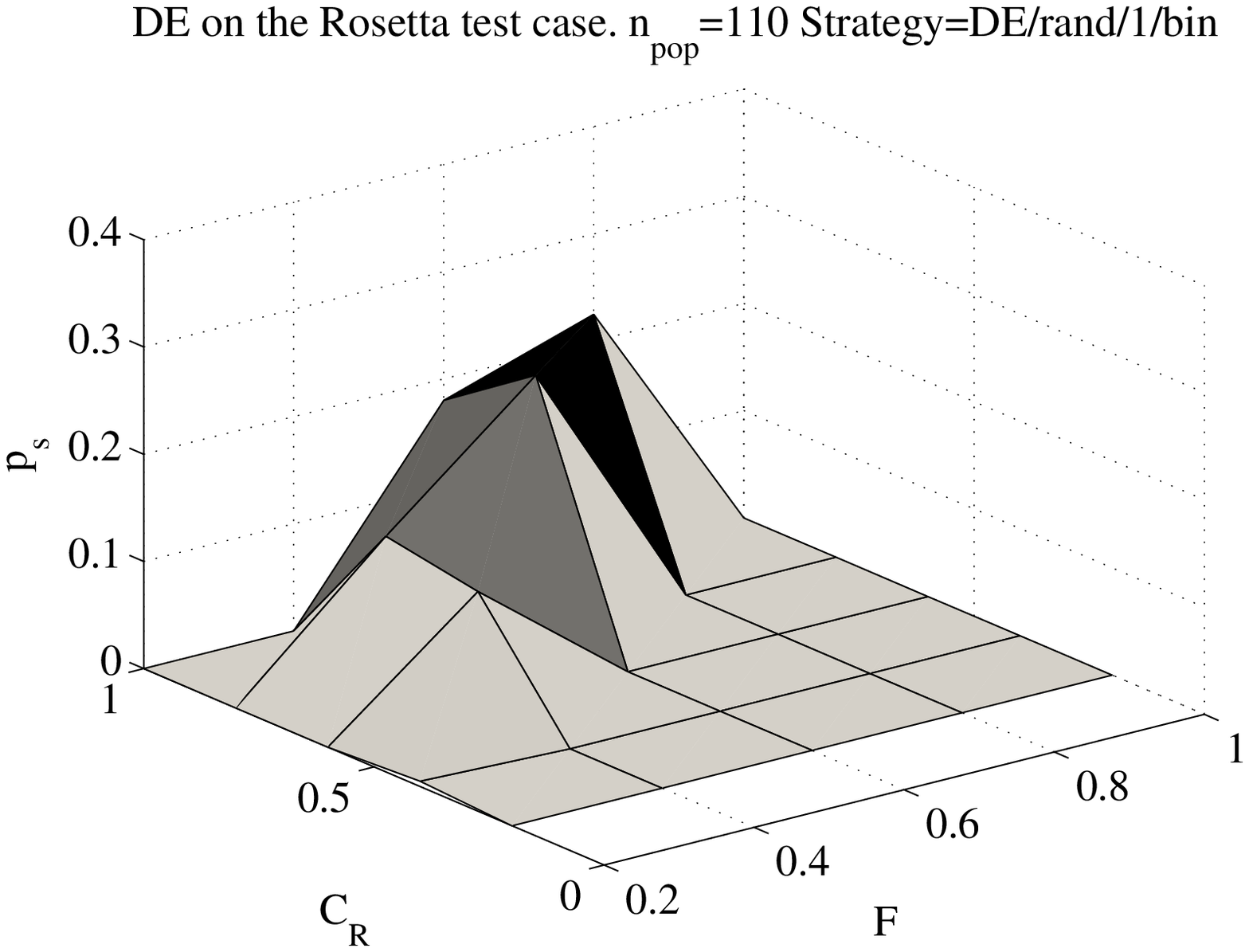} 
 \vskip -5mm
 \caption{Performance of Differential Evolution on the Rosetta test case, with strategy \emph{DE/rand/1/bin} and $n_{pop}$=110} \label{fig:DE_Rosetta_perfo7pop}
\end{center}
\end{figure}

\subsection{Test Results}
The performance of all the algorithms are summarized in Figs.
\ref{fig:Evolution_Cassini_nodsm}, \ref{fig:Evolution_Cassini_dsm},
\ref{fig:Evolution_Rosetta2} and \ref{fig:Evolution_Messenger2}. For
the Cassini case with no DSM's, with results shown in Fig. \ref{fig:Evolution_Cassini_nodsm}, IDEA performs exceptionally well
compared to all the other algorithms providing a success rate over
50\% at 200000 function evaluations. Both versions of DE
exhaust their exploration capabilities quite soon and an increase of
the number of function evaluations does not help as the population
has collapsed to a fixed point. The performance is not nearly as good on the Cassini case with DSM's, as seen in Fig. \ref{fig:Evolution_Cassini_dsm}.

Although IDEA is still better than the other algorithms, and in particular than standard
DE, it is comparable to MBH up to 600000 function evaluations and
achieves a moderate 30\% as best result at 1.25 million function
evaluations. This poor performance seems to be due mainly to
problems of local convergence. Note, however, how the general trend
suggests that IDEA is not stagnating as the success rate is steadily
increasing for an increasing number of function evaluations. On the
contrary DE seems to reach a flat plateau. Even for the Rosetta
case, IDEA displays exceptionally good results with DE5c, DE10c, MBH
and MBH-GR substantially equivalent until $N \leq 800k$. After that,
DE reaches a plateau and stops exploring. Both IDEA and MBH,
instead, show a positive monotonic trend until $N = 1.25M$, but IDEA
outperforms both versions of MBH.

In the Messenger case, all algorithms do not perform particularly well if $N \leq 500k$. Due to
the structure of the problem, IDEA needs more time to converge under
the $tol_{conv}$ threshold and only after many re-sampling
iterations is able to show good performance. On the other hand,
all DE's do not display any significant improvement for more than
$500k$ evaluations and even MBH can be considered equally
ineffective (the difference between DE and MBH is lower then the
confidence interval).

\begin{figure}[tbh]
\begin{center}
\vskip -5mm
 \includegraphics[width=9.5cm]{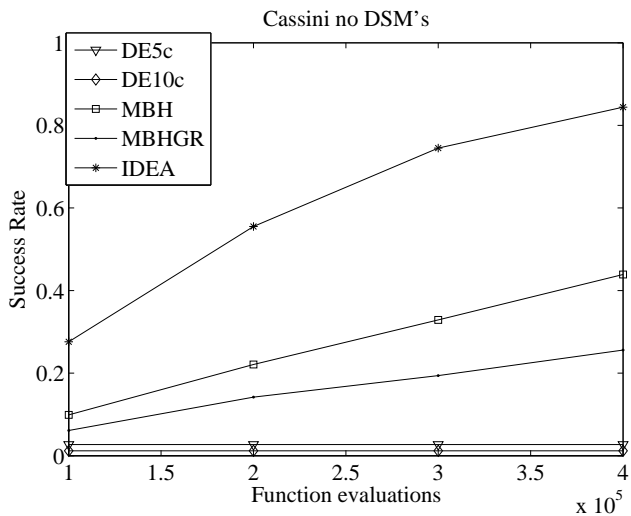}
\vskip -3mm
 \caption{Variation of the success
 rate with the number of function evaluations, for
 the Cassini without DSM's test case} \label{fig:Evolution_Cassini_nodsm}
\end{center}
\end{figure}

\begin{figure}[tbh]
\begin{center}
\vskip -5mm
 \includegraphics[width=9.5cm]{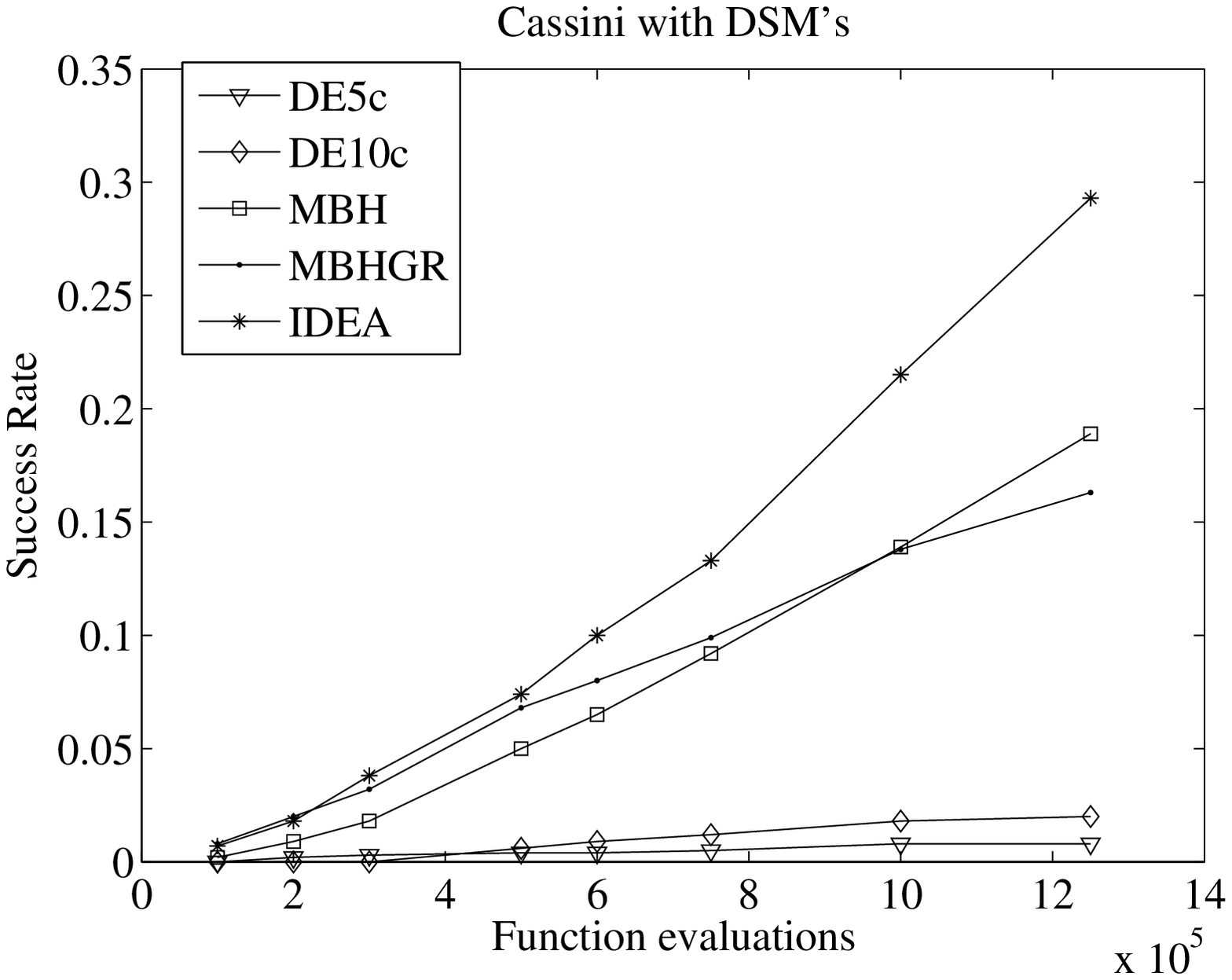}
\vskip -3mm
 \caption{Variation of the success rate with
 the number of function evaluations, for the Cassini with DSM's test case}
 \label{fig:Evolution_Cassini_dsm}
\end{center}
\end{figure}

\begin{figure}[tbh]
\begin{center}
\vskip -5mm
 \includegraphics[width=9.5cm]{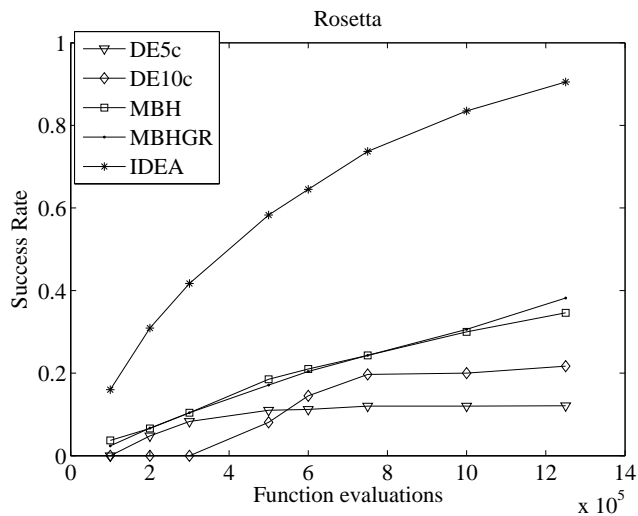}
\vskip -3mm
 \caption{Variation of the success rate with the number of function evaluations, for the Rosetta test case}
 \label{fig:Evolution_Rosetta2}
\end{center}
\end{figure}

\begin{figure}
\begin{center}
\vskip -5mm
 \includegraphics[width=9.5cm]{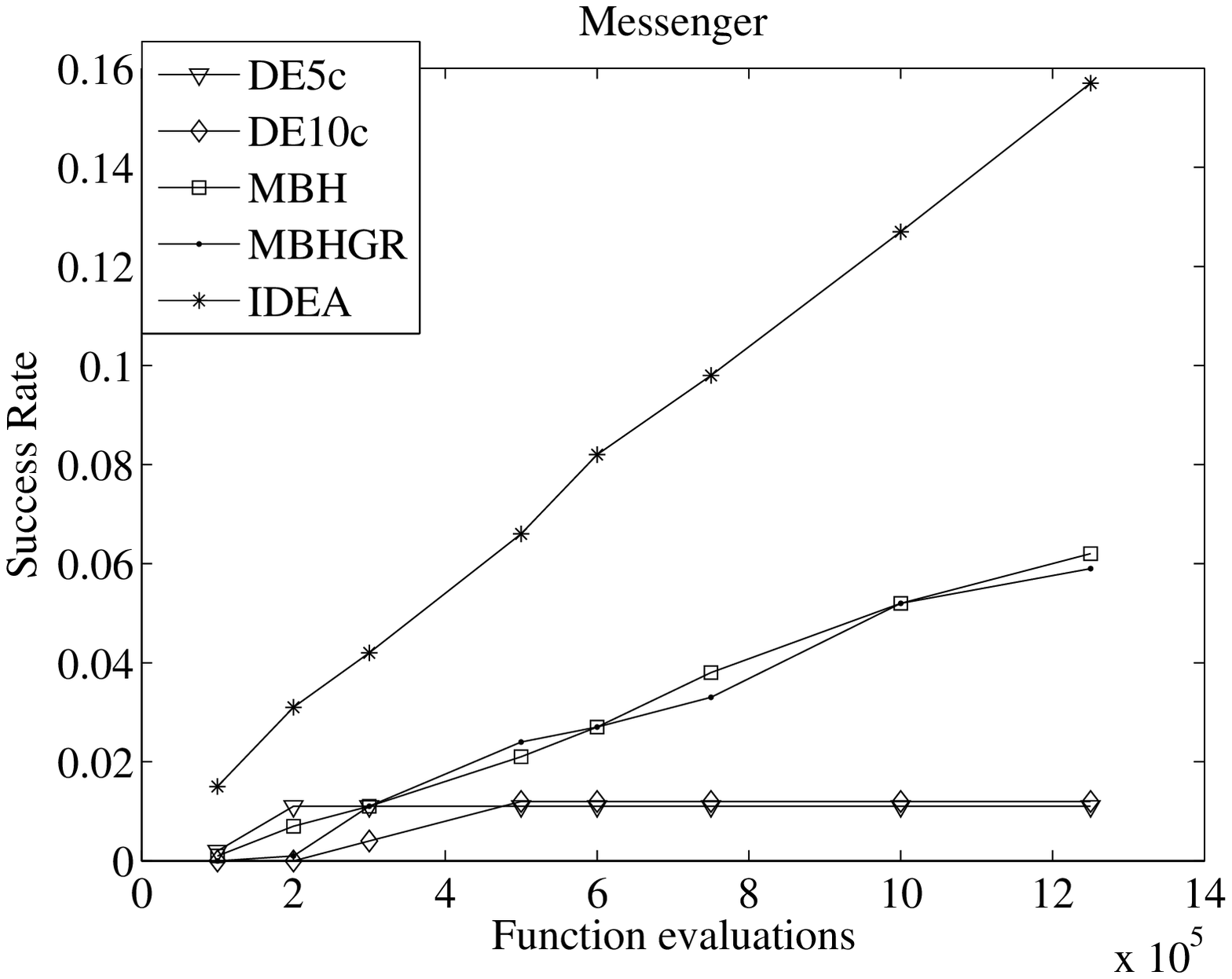}
\vskip -3mm
 \caption{Variation of the success rate with the number of function evaluations, for the Messenger test case}
 \label{fig:Evolution_Messenger2}
\end{center}
\end{figure}

\section{Search Space Analysis}

The different behaviors of IDEA, DE and MBH on the four test cases
can be understood with an analysis of the structure of the search
space. The collection of the results from the tests can be used to
deduce some properties of the problems within the benchmark and to
predict the behavior of the solution algorithms. An understanding of
the characteristics of the benchmark is required to generalize the
result of the tests. In fact, every consideration on the performance
of the algorithms is applicable only to problems with similar
characteristics.

All local minima found in all the tests by the applied global
methods were grouped according to the value of their objective
function. Specifically, the range of values of the objective
function for each test was divided in a finite number of levels,
with each group of minima associated to a particular level.

\begin{figure}[!h]
\begin{center}
\subfigure[\label{fig:1funnel}]{\includegraphics[width=6.9cm]{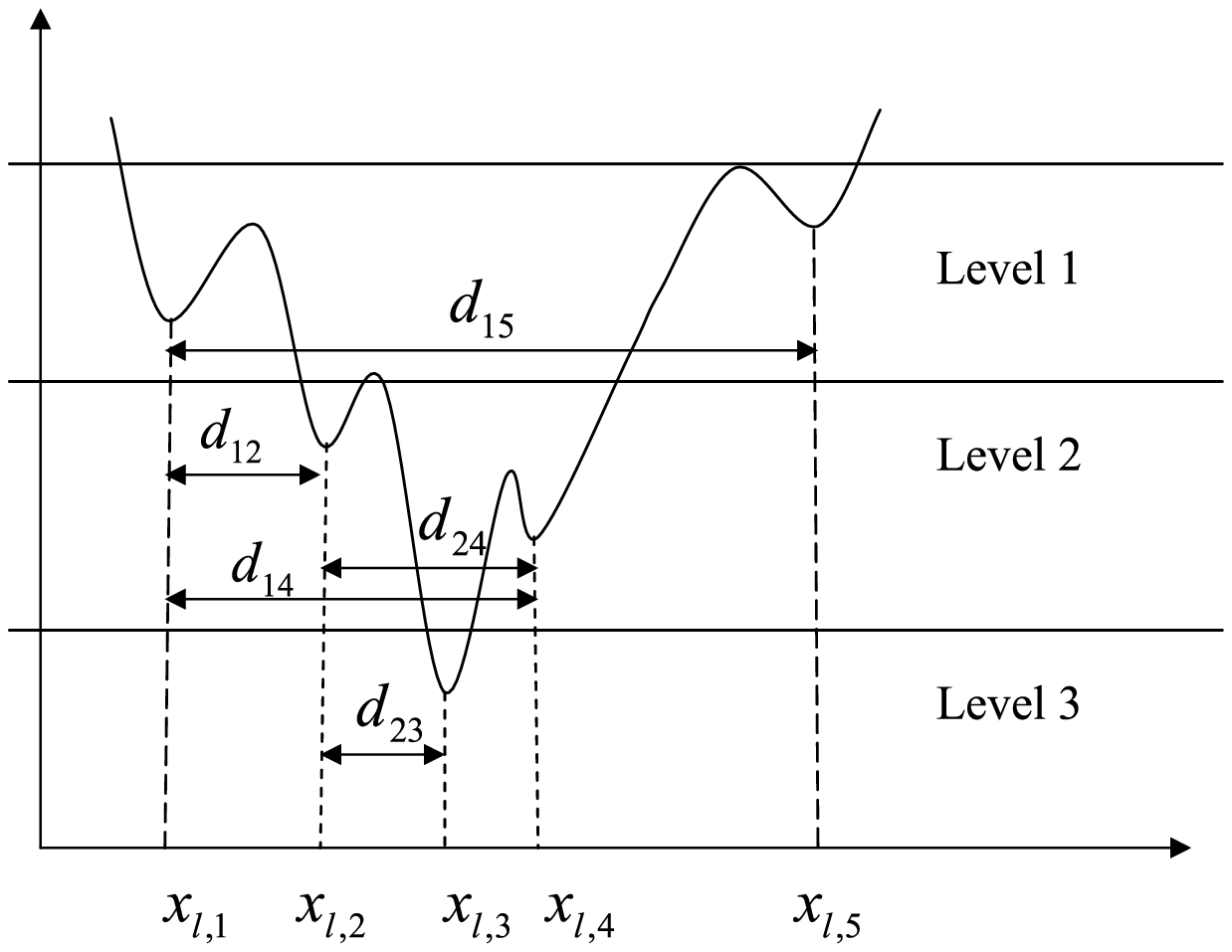}}
\subfigure[
\label{fig:2funnels}]{\includegraphics[width=8.0cm]{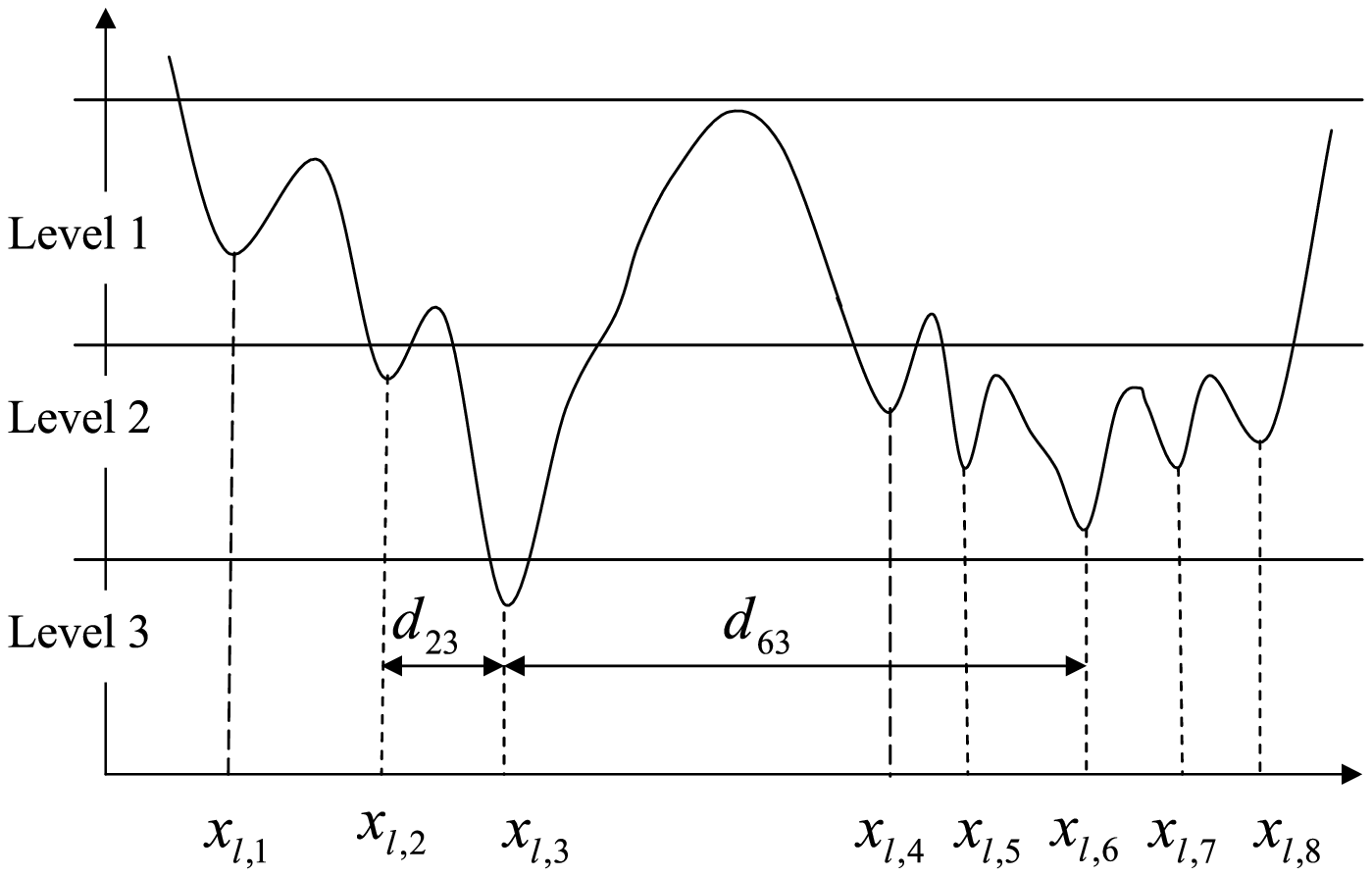}}
\caption{One dimensional example of a) single funnel structure and
b) bi-funnel structure.} \label{fig:funnel1D}
\end{center}
\end{figure}

\begin{figure}
\begin{center}
 \includegraphics[width=8cm]{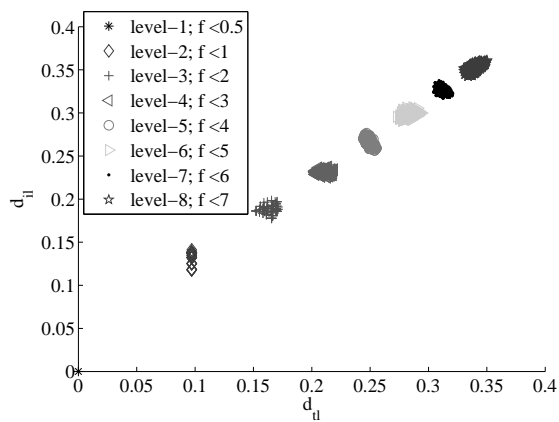}
\vskip -5mm
 \caption{Relative distance of the local minima for the Rastrigin
 function in 5 dimensions} \label{fig:distanze_rastrigin}
\end{center}
\end{figure}

\begin{figure}
\begin{center}
 \includegraphics[width=8cm]{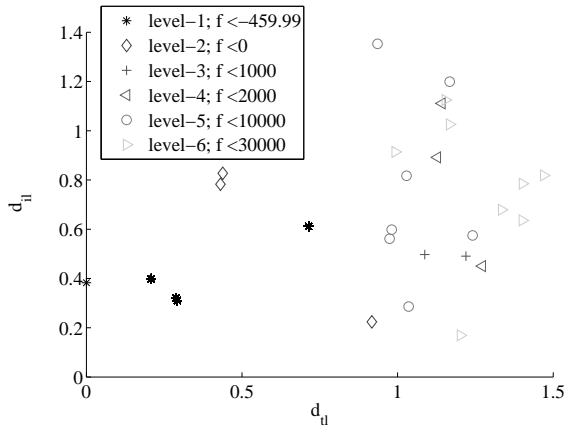}
\vskip -5mm
 \caption{Relative distance of the local minima for
 the Schwefel function in 5 dimensions} \label{fig:distanze_shweffel}
\end{center}
\end{figure}

\begin{figure}[!h]
\begin{center}
\subfigure[\label{fig:Rastrigin1D}]{\includegraphics[width=8.0cm]{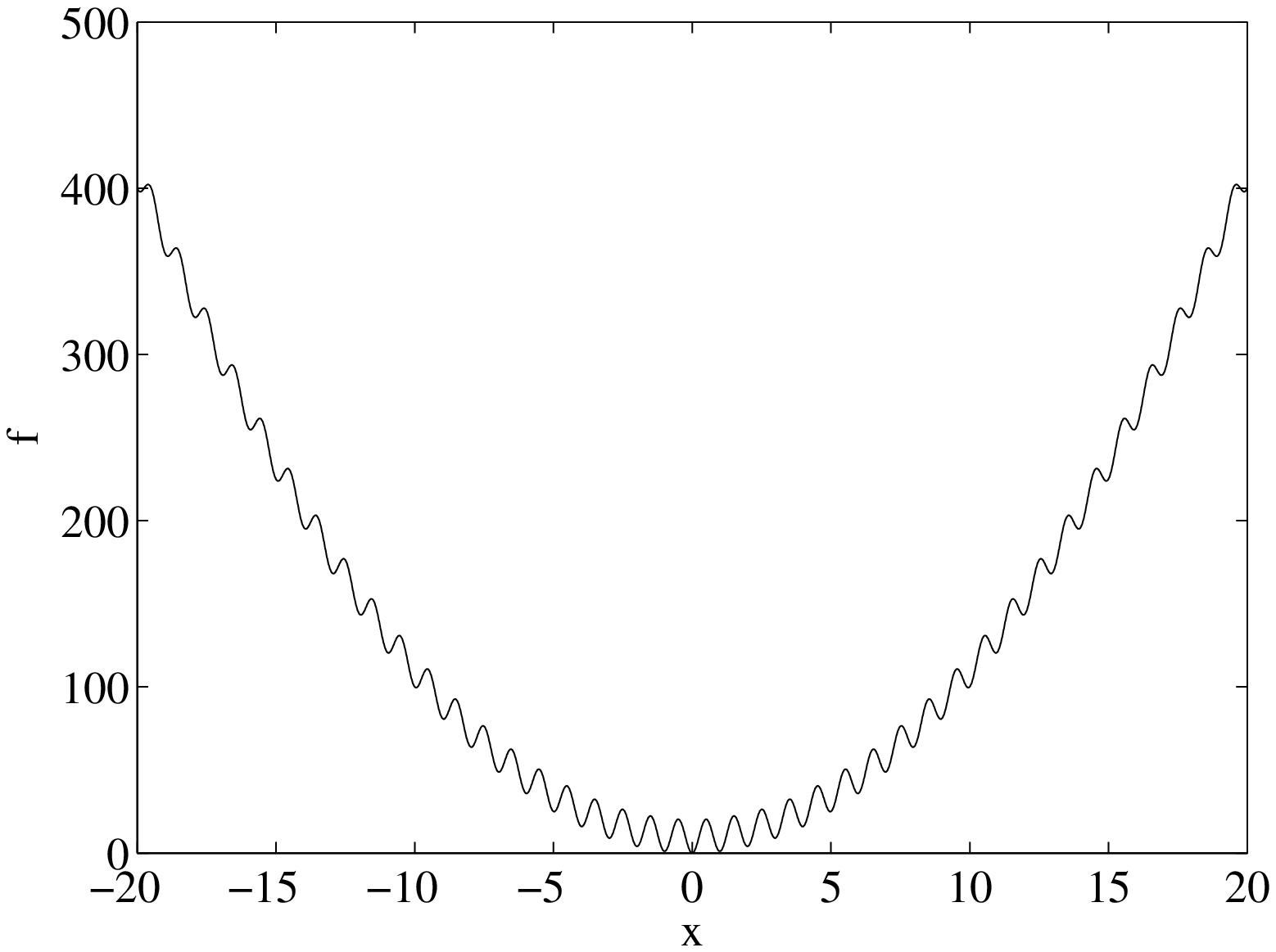}}
\subfigure[
\label{fig:Schweffel1D}]{\includegraphics[width=8.0cm]{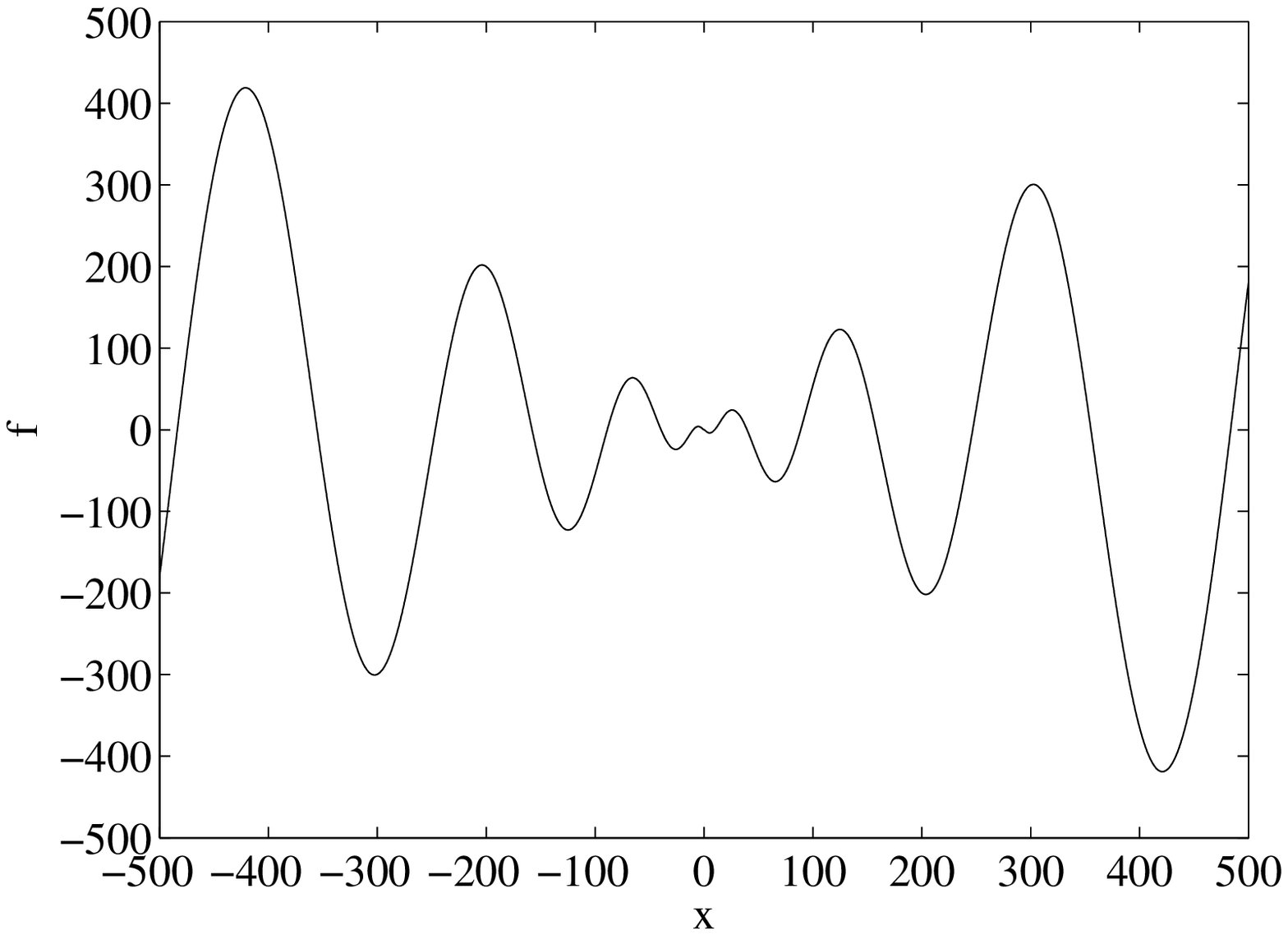}}
\caption{One dimensional example of a) Rastrigin function and b)
Schwefel function.} \label{fig:funnelFunc1D}
\end{center}
\end{figure}

Then, we computed the average value of the relative distance of each
local minimum with respect to all other local minima within the same
level  $d_{il}$ (or intra-level distance), and the average value of the relative distance of each local minimum with
respect to all other local minima in the lower level $d_{tl}$ (or
trans-level distance). The $d_{tl}$ for the lowest level is the
average distance with respect to the best known solution.

The values $d_{il}$ and $d_{tl}$ give an immediate representation of
the diversity of the local minima and the probability of a
transition from one level to another. More precisely, a cluster of
minima with a large intra-level distance and a small trans-level
distance suggests an easy transition to lower values of the
objective function and a possible underlying funnel
structure \cite{Lea:00}. In particular in the case of funnel
structures, the values of $d_{tl}$ and $d_{il}$ should progressively
go to zero. A $d_{il}$ that does not go to zero or clusters with
different values of $d_{tl}$, are the cue to a possible underlying
multi-funnel structure.

Fig.~\ref{fig:funnel1D} provide two illustrative examples. Fig. \ref{fig:1funnel} represents a single
funnel structure with five local minima $x_{l,i}$ where $i=1,...,5$,
and three levels. The intra-level distance at level 2, given by the
distance $d_{24}=x_{l,2}-x_{l,4}$, is lower than
$d_{15}=x_{l,1}-x_{l,5}$, the intra-level distance at level 1. The
same is true for the trans-level distance at level 2, $d_{23}$,
which is lower than the trans-level distance at level 1,
$(d_{12}+d_{14})/2$, for minimum $x_{l,1}$.

Fig. \ref{fig:2funnels}, instead, represents a bi-funnel structure. In this case, the minima
around $x_{l,6}$ have an average intra-level distance lower than
$x_{l,2}$ but a trans-level distance $d_{63}$ much larger than
$d_{23}$. Thus, the two minima $x_{l,2}$ and $x_{l,6}$ will appear
on the $d_{tl}$-$d_{il}$ graph with different values of $d_{il}$ and
$d_{tl}$. If the threshold of level 3 were increased above the
objective value of $x_{l,6}$, then all minima of level 2 would have
similar $d_{tl}$, but the $d_{il}$ at level 3 would not go to zero.

This analysis method is an extension of the work of Reeves and
Yamada \cite{reeves:98}, and is used to concisely visualize the
distribution of the local minima. The definition of the levels
depends on the groups of minima of interest, and can be derived from
mission constraints or by an arbitrary subdivision of the range of
values of the objective function. Different subdivisions reveal
different characteristics of the search space but give only
equivalent cues on the transition probability.

As an example of application of the proposed search space analysis method,
Figs. \ref{fig:distanze_rastrigin} and \ref{fig:distanze_shweffel} show the $d_{il}$-$d_{tl}$ plots for two
well known functions: Rastrigin and Schwefel. In both
cases the dimensionality is 5. The Rastrigin function is known to
have a single funnel and to be globally convex (see the example of
one-dimensional Rastrigin function in Fig. \ref{fig:Rastrigin1D}).
The clusters in the $d_{il}$-$d_{tl}$ plane are aligned along the
diagonal and converge to 0. In the case of the Schwefel function,
which has no single funnel structure (see the example of one-dimensional
Schwefel function in Fig. \ref{fig:Schweffel1D}), the clusters are
more scattered and both the $d_{il}$ and $d_{tl}$ values tend to
remain quite high.

\begin{figure}[!h]
\begin{center}
\subfigure[\label{fig:distance_EVVEJS}]{
\includegraphics[width=8.1cm]{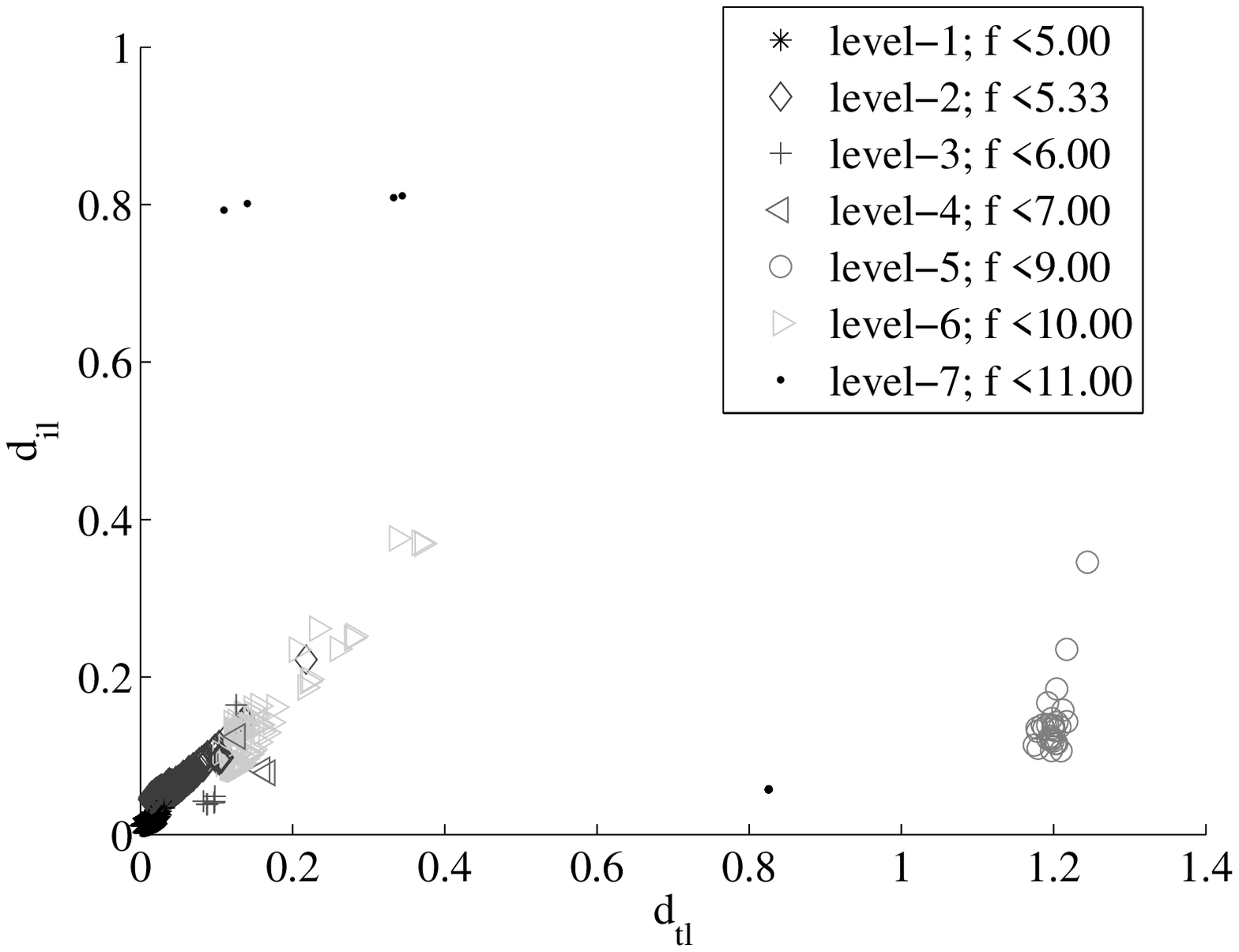}}
\subfigure[\label{fig:distance_EVVEJSDSM}]{
\includegraphics[width=8.cm]{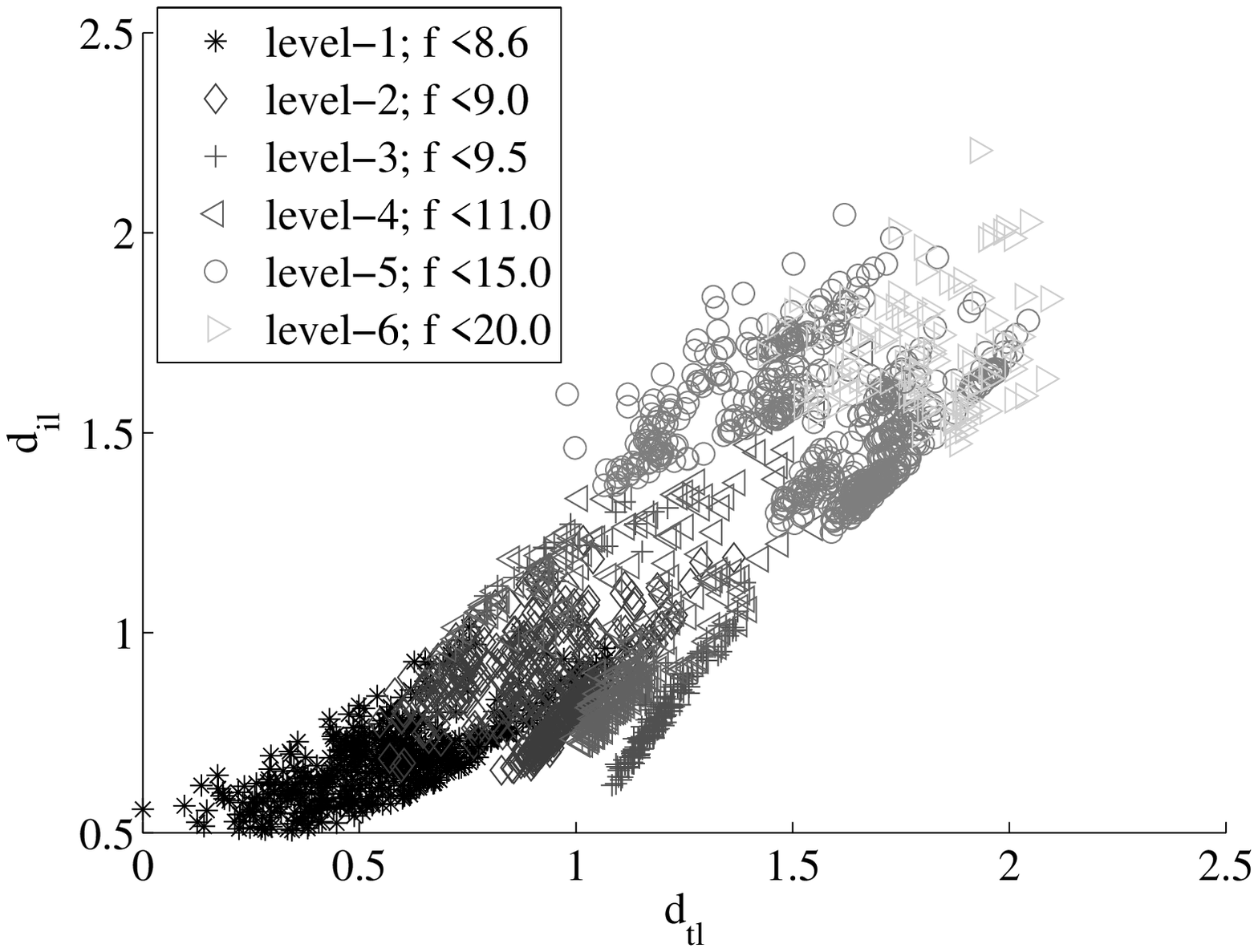}}
 \caption{Relative distance of the local minima for Cassini: a) without DSM's and b) with DSM's.}
\end{center}
\end{figure}

When applied to the benchmark of space trajectory problems, the analysis approach seems to suggest (see Fig.~\ref{fig:distance_EVVEJS}) that the Cassini
case with no DSM's has a structure similar to the one in
Fig.~\ref{fig:2funnels}: the minima at level 5 belong to two
distinct clusters with substantially different $d_{tl}$ and
$d_{il}$. The clusters, corresponding to levels 1, 2 and 3, have
values of $d_{tl}$ and $d_{il}$ both lower than 0.2, which suggests
an easy transition from one level to an another. Thus, below an
objective function of 7.5 km/s there seems to be an underlying
single funnel structure. Note that an easy transition among levels
favors the search mechanism of MBH as demonstrated by the test
results.

Fig.~\ref{fig:distance_EVVEJSDSM} shows that both $d_{tl}$ and
$d_{il}$ progressively tend to zero up to a certain point, after
which $d_{tl}$ goes to zero while $d_{il}$ remains almost unchanged.
The figure suggests that, in the Cassini case with DSM, there is a
 single funnel structure for function values above 9.5 km/s,
while below the minima are scattered and distant from each other.

From Fig. \ref{fig:distanze_Rosetta} we can argue that Rosetta
has a wide zone containing the global minimum together with a number
of local minima with similar cost function. In fact the blue stars
belonging to level 1 are distributed over values of $d_{tl}$ in the
range [0.0, 1.9] and $d_{il}$ in the range [1.2, 1.6]. It should be
noted that from a practical point of view, solutions with a
difference in the total $\Delta v$ of less than 50 m/s are
equivalent, especially in the preliminary design phase of a mission.
Therefore, all the blue stars belonging to level 1 are potentially
good candidates for a space mission. Note that, the solutions at
level 1 are very far apart. Thus, a transition within this group and
the identification of the global minimum can be problematic.

In addition to the solutions at level 1, the search space of Rosetta
presents two interesting groups at level 3 and level 4. These two
sets of local minima have a relatively low inter-distance $d_{il}$
but are far from the global minimum, since $d_{tl}$ is quite high.
They represent two strong attractors that explain why, for example,
the actual Rosetta mission has a total $\Delta v$ a bit higher than
1.7 km/s.

The Messenger mission problem presents a different
structure, see Fig. \ref{fig:distanze_Messenger}. The structure of
the search space of the Messenger test case appears to be
characterized by many small basins separated from one another. Even the level under the threshold adopted to compute the success rate contains two distinct, yet similar, regions.

\begin{figure}
\begin{center}
 \includegraphics[width=8cm]{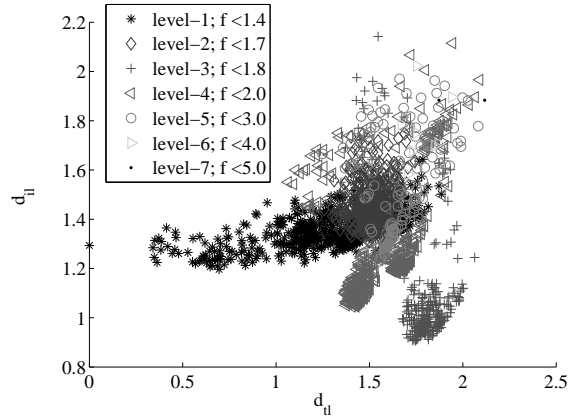}
\vskip -5mm
 \caption{Relative distance of the local minima for Rosetta} \label{fig:distanze_Rosetta}
\end{center}
\end{figure}

\begin{figure}
\begin{center}
 \includegraphics[width=8.0cm]{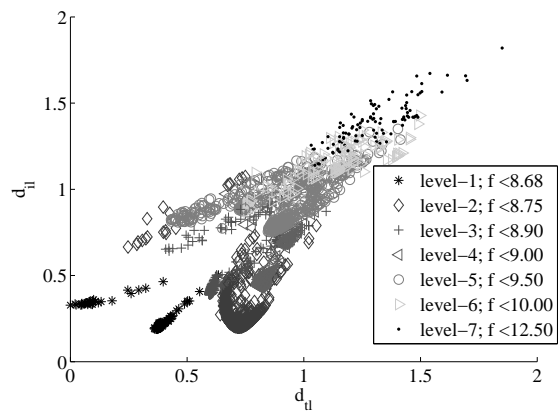}
\vskip -5mm
 \caption{Relative distance of the local minima for Messenger} \label{fig:distanze_Messenger}
\end{center}
\end{figure}

\section{Conclusion}
In this paper we casted an evolutionary heuristic in the form of a discrete map. The discrete map can be seen as a variant of Differential Evolution. We then demonstrated empirically and
theoretically that the map has
a number of fixed points to which it converges asymptotically under
some assumptions on the structure of the search space. This result
suggested the development of an algorithm that outperforms
Differential Evolution on some difficult space trajectory design
problems. The novel algorithm displays a remarkable robustness, i.e.,
the ability to repeatedly converge to solutions with a value of
the cost function close to the best known solution to date. Furthermore, it
shows the desirable characteristic of increasing its performance
with the number of function evaluations, reaching in some cases
success rates which are up to around 25 times higher than the standard DE.
These considerations can be generalized to all problems with similar characteristics of the search space.

\section{Acknowledgments}
The authors would like to thank Dr. Oliver Sch\"{u}tze for his suggestions on some of the theoretical aspects of this work.

\bibliographystyle{IEEEtran}
\bibliography{idea_tec3}

\begin{thebibliography}{10}
\providecommand{\url}[1]{#1}
\csname url@samestyle\endcsname
\providecommand{\newblock}{\relax}
\providecommand{\bibinfo}[2]{#2}
\providecommand{\BIBentrySTDinterwordspacing}{\spaceskip=0pt\relax}
\providecommand{\BIBentryALTinterwordstretchfactor}{4}
\providecommand{\BIBentryALTinterwordspacing}{\spaceskip=\fontdimen2\font plus
\BIBentryALTinterwordstretchfactor\fontdimen3\font minus
  \fontdimen4\font\relax}
\providecommand{\BIBforeignlanguage}[2]{{%
\expandafter\ifx\csname l@#1\endcsname\relax
\typeout{** WARNING: IEEEtran.bst: No hyphenation pattern has been}%
\typeout{** loaded for the language `#1'. Using the pattern for}%
\typeout{** the default language instead.}%
\else
\language=\csname l@#1\endcsname
\fi
#2}}
\providecommand{\BIBdecl}{\relax}
\BIBdecl

\bibitem{clerc:02}
M.~Clerc and J.~Kennedy, ``The particle swarm -- explosion, stability, and
  convergence in a multidimensional complex space,'' \emph{IEEE Transactions on
  Evolutionary Computation}, vol.~6, no.~1, pp. 58 -- 73, 2002.

\bibitem{trlea:03}
I.~C. Trlea, ``The particle swarm optimization algorithm: convergence analysis
  and parameter selection,'' \emph{Information processing letters}, vol.~85,
  pp. 317 -- 325, 2003.

\bibitem{clerc:06}
M.~Clerc, \emph{Particle Swarm Optimization}.\hskip 1em plus 0.5em minus
  0.4em\relax ISTE, 2006.

\bibitem{poli:08}
R.~Poli, ``Dynamics and stability of the sampling distribution of particle
  swarm optimisers via moment analysis,'' \emph{Journal of Artificial Evolution
  and Applications}, 2008.

\bibitem{kenneth:05}
K.~Price, R.~Storn, and J.~Lampinen, \emph{Differential Evolution. A Practical
  Approach to Global Optimization}, ser. Natural Computing Series.\hskip 1em
  plus 0.5em minus 0.4em\relax Springer, 2005.

\bibitem{Dellnitz02}
M.~Dellnitz, O.~Schutze, and S.~Sertl, ``Finding zeros by multilevel
  subdivision techniques,'' \emph{IMA Journal of Numerical Analysys}, vol.~22,
  pp. 167--185, 2002.

\bibitem{bennett:97}
A.~Pr\"{u}gel-Bennett, ``Dynamics and stability of the sampling distribution of
  particle swarm optimisers via moment analysis,'' \emph{Journal of Theoretical
  Biology}, vol. 185, pp. 81--95, Mar 1997.

\bibitem{bennett:01}
A.~Pr\"{u}gel-Bennett and A.~Rogers, ``Modelling {GA} dynamics,'' in
  \emph{Theoretical Aspects of Evolutionary Computing, Natural Computing},
  L.~Kallel, B.~Naudts, and A.~Rogers, Eds.\hskip 1em plus 0.5em minus
  0.4em\relax Springer, 2001, pp. 59--86.

\bibitem{reeves:98}
H.~G. Beyer, ``On the dynamics of {EA}s without seletion,'' in
  \emph{Foundations of Genetic Algorithms 5}, W.~Banzhaf and C.~Reeves,
  Eds.\hskip 1em plus 0.5em minus 0.4em\relax San Francisco, CA: Morgan
  Kaufmann Publisher, 1998.

\bibitem{dell:06}
S.~Sertl and M.~Dellnitz, ``Global optimization using a dynamical systems
  approach,'' \emph{Journal of Global Optimization}, vol.~34, no.~4, pp.
  569--587, 2006.

\bibitem{storn97}
R.~Storn and K.~Price, ``Differential evolution - a simple and efficient
  heuristic for global optimization over continuous spaces,'' \emph{Journal of
  Global Optimization}, vol.~11, pp. 341--359, 1997.

\bibitem{Lea:00}
R.~H. Leary, ``Global optimization on funneling landscapes,'' \emph{Journal of
  Global Optimtisation}, vol.~18, pp. 367--383, 2000.

\bibitem{convex}
R.~T. Rockafellar, \emph{Convex analysis}.\hskip 1em plus 0.5em minus
  0.4em\relax Princeton University Press, 1970.

\bibitem{WalDoy97}
D.~J. Wales and J.~P.~K. Doye, ``Global optimization by basin-hopping and the
  lowest energy structures of lennard-jones clusters containing up to 110
  atoms,'' \emph{Journal of Physical Chemistry A}, vol. 101, pp. 5111--5116,
  July--August 1997.

\bibitem{Loca05}
M.~Locatelli, ``On the multilevel structure of global optimization problems,''
  \emph{Computational Optimization and Applications}, vol.~30, pp. 5--22, 2005.

\bibitem{peng:09}
F.~Peng, K.~Tang, G.~Chen, and X.~Yao, ``Multi-start jade with knowledge
  transfer for numerical optimization,'' in \emph{Proceedings of the 2009 IEEE
  Congress on Evolutionary Computation (CEC2009)}, Trondheim, Norway, May 2009.

\bibitem{neri10}
F.~Neri and V.~Tironnen, ``Recent advances in differential evolution: a survey
  and experimental analysis,'' \emph{Artificial Intelligence Reviews}, vol.~33,
  pp. 61--106, 2010.

\bibitem{auger:05}
A.~Auger and N.~Hansen, ``A restart cma evolution strategy with increasing
  population size,'' in \emph{Proceedings of the 2005 IEEE Congress on
  Evolutionary Computation (CEC2005)}, Edinburgh, Scotland, UK, September 2005.

\bibitem{sentinella:07}
M.~Sentinella, ``Comparison and integrated use of differential evolution and
  genetic algorithms for space trajectory optimisation,'' in \emph{Proceedings
  of the 2007 IEEE Congress on Evolutionary Computation (CEC2007)}, Singapore,
  September 2007.

\bibitem{pinter:84}
J.~Pinter, ``Convergence properties of stochastic optimization procedures,''
  vol.~15, no.~3, pp. 405--427, 1984.

\bibitem{rudolph:96}
G.~Rudolph, ``Convergence of evolutionary algorithms in general search space,''
  in \emph{Proceedings of the IEEE International Conference on Evolutionary
  Computation}, Nagoya, Japan, May 1996.

\bibitem{battin:99}
R.~Battin, \emph{An Introduction to the Mathematics and Methods of
  Astrodynamics}.\hskip 1em plus 0.5em minus 0.4em\relax AIAA, 1999.

\bibitem{becerra:04}
D.~R. Myatt, V.~Becerra, S.~Nasuto, and J.~Bishop, ``Global optimization tools
  for mission analysis and design,'' ESA/ESTEC, Final Rept. ESA Ariadna ITT
  AO4532/18138/04/NL/MV,Call03/4101, 2004.

\bibitem{vasile:06}
M.~Vasile and P.~De~Pascale, ``Preliminary design of multiple gravity-assist
  trajectories,'' \emph{Journal of Spacecraft and Rockets}, vol.~43, no.~4, pp.
  5--22, July--August 2006.

\bibitem{Addis-et-al}
B.~Addis, M.~Locatelli, and F.~Schoen, ``Local optima smoothing for global
  optimization,'' \emph{Optimization Methods and Software}, vol.~20, pp.
  417--437, 2005.

\bibitem{cit:vasileAAS2008}
M.~Vasile, E.~Minisci, and M.~Locatelli, ``Analysis of some global optimization
  algorithms for space trajectory design,'' \emph{AIAA Journal of Spacecraft
  and Rockets}, vol.~47, pp. 334--344, 2010.

\bibitem{adcock:97}
C.~J. Adcock, ``Sample size determination: a review,'' \emph{The Statistician},
  vol.~46, no.~2, pp. 261--283, 1997.

\bibitem{zaharie:02}
D.~Zaharie, ``Critical values for control parameters of differential evolution
  algorithm,'' in \emph{Proceedings of the 2009 IEEE Congress on Evolutionary
  Computation (CEC2009)}, 2002.

\end{thebibliography}

\clearpage
\listoffigures
\clearpage
\listoftables

\end{document}